\documentclass[11pt]{article}

\usepackage[margin=1in]{geometry}
\usepackage{amsmath}
\usepackage{amssymb}
\usepackage{amsthm}
\usepackage[ruled,vlined,linesnumbered]{algorithm2e}\AtBeginDocument{\DontPrintSemicolon}
\usepackage{tikz}
\usetikzlibrary{shapes.geometric}
\usepackage{pgfplots}
\pgfplotsset{compat=1.18}
\usepackage{graphicx}
\usepackage{capt-of}
\usepackage{cite}
\usepackage{float}
\usepackage{hyperref}
\SetArgSty{textnormal}

\definecolor{cutA}{RGB}{214,39,40}
\definecolor{cutB}{RGB}{31,119,180}
\definecolor{cutC}{RGB}{44,160,44}
\definecolor{cutD}{RGB}{23,190,207}
\definecolor{cutE}{RGB}{148,103,189}
\definecolor{cutF}{RGB}{140,86,75}

\newcommand{\br}[1]{\mathopen{}\left( #1 \right)}
\newcommand{\brc}[1]{\mathopen{}\left\{ #1 \right\}}
\newcommand{\spr}[1]{\mathopen{}\left| #1 \right|}
\newcommand{\fl}[1]{\mathopen{}\left\lfloor #1 \right\rfloor}
\newcommand{\cl}[1]{\mathopen{}\left\lceil #1 \right\rceil}

\newcommand{\DPartition}{\textsc{DPartition}}

\newcommand{\bigo}{\mathcal{O}}

\newcommand{\OPTdem}[2]{\OPT_{#2}\br{#1}}

\newcommand{\OPT}{\textup{OPT}}
\DeclareMathOperator*{\argmin}{arg\,min}

\newcommand{\ct}{\text{ct}}

\newcommand{\dist}{\text{dist}}

\newcommand{\vol}{\operatorname{vol}}
\newcommand{\lca}{\operatorname{lca}}

\newcommand{\cC}{\mathcal{C}}

\newcommand{\cF}{\mathcal{F}}

\newcommand{\cT}{\mathcal{T}}

\newtheorem{theorem}{Theorem}[section]
\newtheorem{lemma}[theorem]{Lemma}
\newtheorem{corollary}[theorem]{Corollary}
\theoremstyle{definition}

\usepackage{enumitem}

\title{Graph Partitioning with Demands: Generalized Conductance and its Applications}
\author{Micha{\l} Szyfelbein \\ Gda\'nsk University of Technology \and Dariusz Dereniowski \\ Gda\'nsk University of Technology}
\date{\today}

\begin{document}

\maketitle

\begin{abstract}
  We study various graph partitioning tasks subject to general demands. In \textsc{Graph Partitioning with Demands} (GPD), the goal is to delete minimum cost supply edge set so that each remaining component contains at most a prescribed fraction of the total demand. In the hierarchical analogue of this problem, \textsc{Hierarchical Clustering with Demands} (HCD), the goal is to build a hierarchy of clusters, charging each supply edge by the demand in the smallest cluster containing its endpoints.

  Our main general-graph results are approximation algorithms for GPD and HCD. For every fixed $0<\rho<\eta<1$, we give a polynomial-time bi-criteria algorithm for GPD whose components have demand at most $\eta \cdot W$ and cost at most $\mathcal{O}_{\rho,\eta}(\sqrt{\log n})$ times the optimum cost for threshold~$\rho$. This procedure can then be applied recursively to obtain an $\mathcal{O}(\sqrt{\log n})$-approximation for HCD. The algorithm for GPD itself iteratively uses a crucial subproblem called \textsc{Generalized Conductance}~(GC), which is a new generalization of the unit-demand \textsc{Sparsest Cut}.

  On supply trees, we show that GC can be solved exactly in polynomial time. We also give bi-criteria algorithms for GPD based on anchored-LP rounding and separator recursion, as well as a single-criterion $(4+\epsilon)$-approximation in $n^{O(1/\rho)}$ time. Moreover, through previous reduction, these GPD algorithms yield a $\frac{4e}{e-1}$-approximation for HCD on trees.

  To complement our algorithmic results, assuming SSEH, we prove that HCD is hard to approximate within any constant factor. This also implies that GPD has no constant-factor bi-criteria approximation, for any fixed $0<\rho<\eta<1$ and GC cannot be approximated within any constant factor as well. 
  On trees, assuming the Unique Games Conjecture (UGC), GPD is hard to approximate within a factor of $2-\epsilon$ by reduction from \textsc{Vertex Cover}, while HCD is hard to approximate within a factor of $1.0748$, by reduction from \textsc{Min-Sum Vertex Cover}.
\end{abstract}


\newpage

\section{Introduction}

\textsc{Sparsest Cut} is a central problem among various graph partitioning tasks. Given an instance $G=(V,E,c)$, where $E$ is the set of supply edges and $c\colon E \to \mathbb{N}$ is a capacity function, the goal is to find a subset of vertices $S\subseteq V$ that minimizes

\[
\phi\br{S} = \frac{c\br{\delta_E\br{S}}}{\spr{S}\cdot \spr{\bar{S}}}.
\]
where $\bar{S}=V\setminus S$ and $\delta_E\br{S}=\brc{e\in E\colon \spr{e\cap S}=1}$.

The extensive study of this task \cite{MulticommodityMaxFlowMinCutTheoremsAndTheirUseInDesigningApproximationAlgorithms} has resulted in remarkable connections to topics including metric spaces \cite{TheGeometryOfGraphsAndSomeOfItsAlgorithmicApplications,AnOlogkApproximateMinCutMaxFlowTheoremAndApproximationAlgorithm,OnLipschitzEmbeddingOfFiniteMetricSpacesInHilbertSpace,EuclideanDistortionAndTheSparsestCut} and spectral graph theory \cite{ALowerBoundForTheSmallestEigenvalueOfTheLaplacian,Lambda1IsoperimetricInequalitiesForGraphsAndSuperconcentrators,SpectralPartitioningWorksPlanarGraphsAndFiniteElementMeshes} with the state-of-the-art approximation algorithm due to Arora, Rao, and Vazirani \cite{ExpanderFlowsGeometricEmbeddingsAndGraphPartitioning}, achieving the ratio of $\bigo\br{\sqrt{\log n}}$, where $n=\spr{V}$. Its numerous applications include designing approximation algorithms for other combinatorial problems, where it serves as a building block for various divide-and-conquer strategies. This includes \textsc{Graph Partitioning} where the goal is to find a minimum-capacity subset of edges whose removal splits the graph into small components~\cite{MulticommodityMaxFlowMinCutTheoremsAndTheirUseInDesigningApproximationAlgorithms}. Another example is the \textsc{Hierarchical Clustering} problem, where the goal is to hierarchically cluster data points, so that to minimize the Dasgupta's clustering cost function \cite{ACostFunctionForSimilarityBasedHierarchicalClustering,ApproximateHierarchicalClusteringViaSparsestCutAndSpreadingMetrics,HierarchicalClusteringObjectiveFunctionsAndAlgorithms}. This objective promotes clusters with both balanced-sized children and low capacity of edges between them. In particular, given an $\alpha$-approximation algorithm for the \textsc{Sparsest Cut}, one obtains an $\bigo\br{\alpha}$-approximation algorithm for both of the aforementioned problems. 
Our main topic of interest is whether one could establish similar relationships for the case when both \textsc{Graph Partitioning} and \textsc{Hierarchical Clustering} are subject to demands over pairs of vertices. 
Our question then naturally extends to asking which particular generalization of the \textsc{Sparsest Cut} would be appropriate for such needs.

In order to answer these questions we firstly examine which property of \textsc{Sparsest Cut} allows approximating other balanced partitioning tasks in the first place.
This usefulness of \textsc{Sparsest Cut} can be attributed to the following observation, implicitly exploited by many algorithms \cite{MulticommodityMaxFlowMinCutTheoremsAndTheirUseInDesigningApproximationAlgorithms,AnOlogkApproximateMinCutMaxFlowTheoremAndApproximationAlgorithm,ACostFunctionForSimilarityBasedHierarchicalClustering,ApproximateHierarchicalClusteringViaSparsestCutAndSpreadingMetrics,HierarchicalClusteringObjectiveFunctionsAndAlgorithms,ImprovedApproximationAlgorithmsForMinimumWeightVertexSeparators}.
The quantity $\spr{S}\cdot \spr{\bar{S}}$ can be interpreted in two separate ways. First, imagine $D=\binom{V}{2}$ to be a set of demand edges where each edge has a unit demand. Then, $\spr{S}\cdot \spr{\bar{S}}$ is the total demand crossing the cut. This interpretation is useful for the sake of developing approximation algorithms for the \textsc{Sparsest Cut} itself, since a weak duality between \textsc{Sparsest Cut} and multicommodity flow type problems has been established \cite{AnOlogkApproximateMinCutMaxFlowTheoremAndApproximationAlgorithm,MulticommodityMaxFlowMinCutTheoremsAndTheirUseInDesigningApproximationAlgorithms}. Secondly, $\spr{S}\cdot \spr{\bar{S}}$ can be observed to be large when both partitions are of roughly equal size. This is useful when developing algorithms for other graph problems \cite{ExpanderFlowsGeometricEmbeddingsAndGraphPartitioning,ImprovedApproximationAlgorithmsForMinimumWeightVertexSeparators,MulticommodityMaxFlowMinCutTheoremsAndTheirUseInDesigningApproximationAlgorithms}, since the product in the denominator becomes a measure of the balance of the cut.

However, if one wished to generalize the \textsc{Sparsest Cut} to an arbitrary demand setting, these two interpretations would no longer be equivalent. Let now $D\subseteq\binom{V}{2}$ be an undirected demand graph and let $w\colon D\to \mathbb{N}$ be a demand function.
Traditionally, \textsc{Generalized Sparsest Cut} aims at finding a subset of vertices $S\subseteq V$ that minimizes
\[\phi_w\br{S} = \frac{c\br{\delta_E\br{S}}}{w\br{\delta_D\br{S}}}.
\]
This problem is also well-understood, and the state-of-the-art approximation algorithm by Arora, Lee, and Naor \cite{EuclideanDistortionAndTheSparsestCut} achieves $\bigo\br{\sqrt{\log n}\cdot \log \log n}$-approximation ratio. However, as discussed above, in order to use sparsity in a wider range of applications, it would be desirable to approximate a quantity measuring the balance of a cut.
To the best of our knowledge, no such results have been obtained in the literature yet. Moreover, the above quantity may assign very large value to a cut $(S, \bar{S})$ even if it is very balanced.
Consider as a toy example a cut where both partitions have internal demand $w\br{S}=w\br{\bar{S}}=w\br{V}/2$ but the demand between them is $w\br{\delta_D\br{S}} = 0$. Then, (by our convention) $\phi_w\br{S}=\infty$, even though the cut is perfectly balanced. This motivates us to focus on a different graph parameter.
The graph \emph{conductance} is a well-known measure of the balance of a cut, defined as
\[
\psi\br{S} = \frac{c\br{\delta_E\br{S}}}{\spr{\delta_D\br{S}}\cdot \spr{\delta_D\br{\bar{S}}}}.
\]
where $\delta_D[S]=\brc{e\in D\colon e\cap S\neq \emptyset}$\footnote{There are several similar definitions of the conductance, however they are all equivalent up to constant factors.}. In fact, for the setup of unit demands, minimizing the conductance is equivalent to the \textsc{Sparsest Cut}, up to a constant factor. However, for the general demands this is not the case, and this will be useful in our applications. The \emph{generalized conductance} of a cut $\br{S,\bar{S}}$ will be defined as
\[
\psi_w\br{S} = \frac{c\br{\delta_E\br{S}}}{w\br{S, V}\cdot w\br{\bar{S}, V}},
\]
where $w\br{S,V}=w\br{S}+w\br{\delta_D\br{S}}$.
We call the problem of minimizing $\psi_w$ \textsc{Generalized Conductance} (GC). The goal of this work is to examine whether this new notion can be leveraged to approximate other graph partitioning tasks with general demands.

We now define the problems we are interested in more formally.
For $U\subseteq V$, let $w(U)$ denote the total weight of demand pairs with both endpoints in $U$, and set $W=w(V)$. In \textsc{Graph Partitioning with Demands} (GPD), given a target $\rho\in(0,1)$, the goal is to delete a minimum-cost set of supply edges $C$ so that every resulting component $H\in G\setminus C$ (where $H=(V_H,E_H)$) satisfies $w(V_H)\leq\rho W$. In \textsc{Hierarchical Clustering with Demands} (HCD), the goal is to find a rooted tree $\mathcal T$ with leaves $V$ minimizing $c\br{\cT}=\sum_{uv\in E} c(uv)\cdot w\br{V_{\lca_{\mathcal T}(u,v)}}$.
Here $V_x$ denotes the leaves below node $x$, so each supply edge is charged by the demand inside the smallest cluster containing its endpoints.

\subsection{Our results and techniques}

Table~\ref{tab:main-results} summarizes our upper and lower approximation bounds.
\begin{table}[t]
\centering
\small
\setlength{\tabcolsep}{2pt}
\renewcommand{\arraystretch}{1.08}
\begin{tabular}{p{0.10\linewidth}p{0.16\linewidth}p{0.08\linewidth}p{0.40\linewidth}p{0.20\linewidth}}
\hline
Problem & Supply graph & Type & Bound & Reference \\
\hline
\textsc{GPD} & General & Upper & $\bigo_{\rho,\eta}(\sqrt{\log n})^{*}$ & \hyperref[cor:graph-partitioning-approximation]{Corollary~\ref*{cor:graph-partitioning-approximation}} \\
\textsc{GPD} & Tree & Upper & $\min\brc{\frac{2\eta}{\eta-\rho},\left\lceil\ln\left(\frac{1-\rho}{\eta-\rho}\right)\right\rceil}^{*}$ & \hyperref[thm:tree-anchored-rounding]{Theorems~\ref*{thm:tree-anchored-rounding}}, \hyperref[thm:tree-centroid-max-coverage]{\ref*{thm:tree-centroid-max-coverage}} \\
\textsc{GPD} & Tree & Upper & $4+\varepsilon$ & \hyperref[thm:tree-centroid-full-quota]{Theorem~\ref*{thm:tree-centroid-full-quota}} \\
\textsc{HCD} & General & Upper & $\bigo(\sqrt{\log n})$ & \hyperref[cor:hc-general-bounds]{Corollary~\ref*{cor:hc-general-bounds}} \\
\textsc{HCD} & Tree & Upper & $\frac{4e}{e-1}$ & \hyperref[cor:tree-hc-improved]{Corollary~\ref*{cor:tree-hc-improved}} \\
\textsc{GC} & General & Upper & $\bigo(\sqrt{\log n})$ & \hyperref[thm:generalized-conductance-approx]{Theorem~\ref*{thm:generalized-conductance-approx}} \\
\textsc{GC} & Tree & Upper & $1$ (exact) & \hyperref[thm:tree-generalized-conductance]{Theorem~\ref*{thm:tree-generalized-conductance}} \\
\textsc{GPD} & General & Lower & No constant factor$^{*}$ (SSEH) & \hyperref[cor:sse-gpd-bicriteria-hardness]{Corollary~\ref*{cor:sse-gpd-bicriteria-hardness}} \\
\textsc{GPD} & Tree & Lower & $2-\varepsilon$ (UGC) & \hyperref[thm:tree-gpd-vertex-cover-hardness]{Theorem~\ref*{thm:tree-gpd-vertex-cover-hardness}} \\
\textsc{HCD} & General & Lower & No constant factor (SSEH) & \hyperref[thm:sse-hcd-hardness]{Theorem~\ref*{thm:sse-hcd-hardness}} \\
\textsc{HCD} & Tree & Lower & $1.0748$ (UGC) & \hyperref[thm:hc-star-hardness]{Theorem~\ref*{thm:hc-star-hardness}} \\
\textsc{GC} & General & Lower & No constant factor (SC hardness) & \hyperref[cor:generalized-conductance-hardness]{Corollary~\ref*{cor:generalized-conductance-hardness}} \\
\hline
\end{tabular}
\caption{Our approximation and hardness results. A superscript $^*$ marks bi-criteria bounds.}
\label{tab:main-results}
\end{table}

The first technical section, after preliminaries, studies the \textsc{Generalized Conductance}, which is our main tool for tackling other partitioning problems on general graphs. We show through a simple reduction that we can encode this problem into an instance of \textsc{Min-Ratio Vertex Cut}. This gives an $\mathcal{O}(\sqrt{\log n})$ approximation, using the work of of Feige, Hajiaghayi, and Lee \cite{ImprovedApproximationAlgorithmsForMinimumWeightVertexSeparators}. Moreover, 
we show that on supply trees, some optimal cut separates exactly one supply edge, implying an exact polynomial-time algorithm. 

Then, we move towards our main algorithmic results for general graphs. We show how to turn any $\gamma$-approximate procedure for GC into a bi-criteria approximation algorithm for GPD. The algorithm repeatedly cuts the remaining components using the cut found for GC, until all pieces are of small enough size.
 Such a procedure is relatively standard in many balanced partitioning applications; however, in our case the generalized demand function requires more careful analysis. Importantly, this reduction only works for a limited range of parameters. In order to extend our result to the remaining parameter values, we apply this procedure recursively, until all components have appropriate size. Via a careful analysis we then show that the depth of this recursion only depends on choice of $\rho$ and $\eta$, which gives an $\mathcal{O}_{\rho,\eta}(\gamma)$-approximation for all fixed $\rho<\eta$. Combining it with the GC approximation yields an $\mathcal{O}_{\rho,\eta}(\sqrt{\log n})$-guarantee for GPD.
After that, we also show how to apply any GPD algorithm to give an approximate solution for HCD. To do so, we recursively apply it to refine clusters until all of them become singletons. This loses only a constant factor in the approximation ratio, yielding an $\mathcal{O}(\sqrt{\log n})$-approximation for HCD.

Then, we move towards GPD on supply trees. Although a constant-factor bi-criteria approximation for this setup is implied by the previous reductions, the constant guaranteed by this algorithm is relatively large. We therefore aim to provide lower approximation guarantees for this case.
In order to address this issue we turn to techniques independent of the \textsc{Generalized Conductance}. Firstly, we give an
anchored-LP relaxation with a simple and elegant threshold rounding with a bi-criteria factor $2\eta/(\eta-\rho)$ guarantee. Then, we give a combinatorial algorithm based on using small vertex separators in order to recursively partition the tree, which gives a bi-criteria factor of $\left\lceil\ln\left((1-\rho)/(\eta-\rho)\right)\right\rceil$. En route we use a $(1-1/e)$-approximation for \textsc{Maximum Coverage} to reduce excess demand at each recursion level by a factor of $e$. Applying this procedure as a subroutine in the top-down HCD construction also gives a $4e/(e-1)$-approximation for HCD on trees. A slight modification of the \textsc{Maximum Coverage} based procedure which uses the \textsc{Partial Multicut} instead, gives a true single-criterion $(4+\varepsilon)$-approximation for GPD in $n^{\bigo(1/\rho)}$ time.

In Section~\ref{sec:gpd-hardness}, assuming the Small-Set Expansion Hypothesis, we show that HCD has no constant-factor approximation. This implies that GPD has no constant-factor bi-criteria approximation for any fixed $0<\rho<\eta<1$ on general graphs, and furthermore, GC has no constant factor approximation as well. Moreover, assuming UGC, we show that GPD and HCD on trees cannot be approximated within factors $2-\varepsilon$ and $1.0748$, respectively.


\subsection{Related work}

The modern theory of \textsc{Sparsest Cut} starts from approximate max-multicommodity flow/min-sparsest cut duality. By exploiting this duality, Leighton and Rao provided an $\bigo\br{\log n}$-approximation \cite{MulticommodityMaxFlowMinCutTheoremsAndTheirUseInDesigningApproximationAlgorithms}. Aumann and Rabani refine this to $\bigo\br{\log k}$ when only $k$ demand edges have nonzero weight \cite{AnOlogkApproximateMinCutMaxFlowTheoremAndApproximationAlgorithm}. Arora-Rao-Vazirani finally improve it by exhibiting an $\bigo\br{\sqrt{\log n}}$ approximation for \textsc{Uniform Sparsest Cut} \cite{ExpanderFlowsGeometricEmbeddingsAndGraphPartitioning}. Obtaining these results required a new approach based on semidefinite programming and high-dimension geometry, which was then further refined by Arora, Lee, and Naor to obtain an $\bigo\br{\sqrt{\log n}\cdot \log\log n}$-approximation for general demands \cite{EuclideanDistortionAndTheSparsestCut}. Classical UGC-based hardness results rule out constant-factor approximations for \textsc{Multicut} and arbitrary-demand \textsc{Sparsest Cut} \cite{OnTheHardnessOfApproximatingMulticutAndSparsestCut,TheUniqueGamesConjectureIntegralityGapForCutProblemsAndEmbeddabilityOfNegativeTypeMetricsIntoL1}. Also, related work connects UGC to Small-Set Expansion \cite{OnTheHardnessOfApproximatingMulticutAndSparsestCut,TheUniqueGamesConjectureIntegralityGapForCutProblemsAndEmbeddabilityOfNegativeTypeMetricsIntoL1,GraphExpansionAndTheUniqueGamesConjecture}. A recent OpenAI preprint claims to prove UGC~\cite{OAI:The-Unique-Games-Theorem-September-23-2026}, if correct, it would remove the UGC assumption from these lower bounds. Separately, an OpenAI preprint claims constant-factor hardness for uniform-demand \textsc{Sparsest Cut}~\cite{OAI:Constant-factor-hardness-of-uniform-sparsest-cut-September-24-2026} assuming $P\neq NP$.

\textsc{Balanced Graph Partitioning} asks to subdivide a graph into $k$ clusters, each of size roughly $n/k$ so as to minimize the cost of the cut edges. Krauthgamer, Naor, and Schwartz study \textsc{$k$-Balanced Partitioning} and give a bi-criteria approximation ratio of $\bigo\br{\sqrt{\log n\cdot \log k}}$ \cite{PartitioningGraphsIntoBalancedComponents}. A related separator `machinery' is developed by Feige, Hajiaghayi, and Lee, who obtain an $\bigo\br{\sqrt{\log n}}$ approximation for \textsc{Min-Ratio Vertex Cut}, a vertex analogue of the \textsc{Sparsest Cut} which gives an $\bigo\br{\sqrt{\log n}}$ bi-criteria approximation for \textsc{Vertex Separators} in general graphs \cite{ImprovedApproximationAlgorithmsForMinimumWeightVertexSeparators}. 
A hierarchical analogue of \textsc{Graph Partitioning} is the problem of \textsc{Hierarchical Clustering}, where the goal is to recursively partition the graph into more refined clusters until all of them become singletons.
Dasgupta introduced the now-standard similarity-based objective and analyzed a top-down algorithm based on repeated \textsc{Sparsest Cut}s, obtaining an $\bigo\br{\alpha\cdot\log n}$ guarantee when \textsc{Uniform Sparsest Cut} admits an $\alpha$-approximation \cite{ACostFunctionForSimilarityBasedHierarchicalClustering}. Charikar and Chatziafratis tightened this analysis to $\bigo\br{\alpha}$ (thus $\bigo\br{\sqrt{\log n}}$ using ARV) \cite{ApproximateHierarchicalClusteringViaSparsestCutAndSpreadingMetrics}. Authors in \cite{HierarchicalClusteringObjectiveFunctionsAndAlgorithms} obtain similar results but using a different argument. Recently, an $\bigo\br{\sqrt{\log n}}$-approximation was also obtained for the vertex variant of \textsc{Hierarchical Clustering} as well \cite{ApproximatingTheAverageCaseGraphSearchProblemWithNonUniformCosts}.

\section{Preliminaries}

Most instances of the problems considered in this work consist of a graph $G=\br{V,E,D,c,w}$, where $V$ is the vertex set, $E$ is the set of \emph{supply edges}, $D$ is the set of \emph{demand edges}, $c\colon E \to \mathbb{N}$ is the \emph{capacity} function which encodes the cost of cutting a supply edge, and $w\colon D\to \mathbb{N}$ is the \emph{demand} function which assigns a non-negative weight to each demand edge.
For vertices $u,v\in V$ sharing an edge in $E$, whenever not explicitly stated, we denote by $uv$ the supply edge between $u$ and $v$.
For any $C\subseteq E$, $G\setminus C$ is the graph obtained by removing the edges in $C$ from $G$, and with a slight abuse of notation, we also treat it as the set of connected components obtained by removing the edges in $C$ from $G$.
We write $\br{V_H,E_H,D_H,c_H,w_H}$ to point out that given parameters represent a particular graph or subgraph $H$ different than $G$.
For $U\subseteq V$, the induced subgraph $G[U]$ has the vertex set $U$, the supply-edge set $E[U]=\brc{uv\in E:u,v\in U}$, the demand-edge set $D[U]=\brc{uv\in D:u,v\in U}$, and the restrictions of $c$ and $w$ to these sets.

Fix an edge set $X\subseteq\binom{V}{2}$ (usually either $X=E$ or $X=D$).
Whenever we refer to a function $f\colon X\to \mathbb{N}$, if $X\subseteq E$, then $f=c$ and if $X\subseteq D$, then $f=w$.
For a pair of disjoint subsets $A, B\subseteq V$, we define $f\br{A}=\sum_{uv\in X[A]}f\br{uv}$ and $f\br{A,B}=\sum_{u\in A,v\in B\colon uv\in X}f\br{uv}$.
Whenever $A, B\subseteq V$ are not disjoint, we define $f\br{A,B}=f\br{A}+f\br{A,B\setminus A}$.
Denote by $W=w\br{V}=\sum_{uv\in D}w\br{uv}$ the total value of the demand in $G$.
For $S\subseteq V$, $\delta_{X}(S)=\brc{uv\in X\colon\spr{\brc{u,v}\cap S}=1}$ consists of the edges in $X$ with exactly one endpoint in $S$.
For brevity $\delta_{X}(v)=\delta_{X}(\{v\})$, $v\in V$. The \emph{weighted degree} of $u$ and the \emph{volume} of the vertices in $S$ are respectively $
	d_X(u)=f\br{\delta_X(u)}$ and $
	\operatorname{vol}_X(S)=\sum_{u\in S}d_X(u).$
We have
\begin{equation}\label{eq:demand-degree-identity}
	\vol_X(S)
	=2f\br{S}+f\br{S,V\setminus S}.
\end{equation}
For $A\subseteq V$ and a set of edges $X$, let $\delta_X[A]=\brc{uv\in X:\brc{u,v}\cap A\neq\emptyset}=X[A]\cup\delta_X(A)$ denote the edges of $X$ incident to $A$.
For parameters $\rho,\eta\in(0,1)$, write $B_\rho=\rho W$ and $B_\eta=\eta W$ for the corresponding demand budgets.
For a subset $C\subseteq X$, $f\br{C}=\sum_{e\in C}f\br{e}$ (recall that if $X=D$, then $f=w$ and if $X=E$, then $f=c$).
For a subset $S\subseteq V$, we denote by $\bar{S}$ the complement of $S$ in $V$, i.~e. $\bar{S}=V\setminus S$.
A \emph{cut} is any subset $S\subseteq V$ such that $S\neq \emptyset$, $S\neq V$ and we denote it either by just $S$ or we write $\br{S, \bar{S}}$ to explicitly indicate its two sides.

Whenever an objective function is given as a ratio and its denominator equals $0$, we set the value of this objective to $\infty$.

\section{Warm Up: Generalized Conductance}

We start with relatively simple results regarding \textsc{Generalized Conductance}. We begin by proving that in an easier setup, when the supply graph is a tree, an optimal cut is always obtained by removing a single supply edge. This immediately yields an exact polynomial-time algorithm that just checks all supply edges and picks the best cut.
Secondly, for general graphs we use a reduction to the \textsc{Min-Ratio Vertex Cut} problem, which yields a general $\bigo\br{\sqrt{\log n}}$-approximation algorithm using known results.

\subsection{Supply trees}

\begin{theorem}\label{thm:tree-generalized-conductance}
  If the supply graph $T$ is a tree, an optimal \textsc{Generalized Conductance} cut is obtained by removing one supply edge. Consequently, the problem admits an exact polynomial-time algorithm on supply trees.
\end{theorem}
\begin{proof}
  If every cut has infinite value or $w\br{V}=0$, the claim follows immediately. Otherwise, let $\br{S^*,\bar{S}^*}$ be an optimal cut and for brevity denote $
    W^*=w\br{S^*,V}\cdot w\br{\bar{S}^*,V}$ and $ E^*=\delta_E\br{S^*}.$
  For every $e\in E^*$, let $S_e$ and $\bar{S}_e$ be the vertex sets of the two components of $T\setminus e$, and denote $
    W_e=w\br{S_e,V}\cdot w\br{\bar{S}_e,V}$. We want to show that there exists an edge $e\in E^*$ with $\psi_{w}\br{S_e}\leq \psi_{w}\br{S^*}$.
    To this end we employ the following lemma.
  \begin{lemma}\label{lemma:tree_edge_weight_sum}
    It holds $\sum_{e\in E^*}W_e\geq W^*$.
  \end{lemma}
  \begin{proof}
    First, for any $A\subseteq V$, we have $w\br{A,V}=\sum_{f\in\delta_D[A]}w\br{f}$.
    For a demand edge $f$, define $I_A(f)=1$ if $f\in\delta_D[A]$ and $I_A(f)=0$ otherwise. For two demand edges $f,g\in D$, possibly with $f=g$, set $q_A(f,g)=I_A(f)\cdot I_{\bar A}(g)+I_A(g)\cdot I_{\bar A}(f)$.
    This gives
    \begin{equation}\label{eq:tree-symmetrized-demand-pairs}
      w\br{A,V}\cdot w\br{\bar A,V}
      =\frac{1}{2}\sum_{f,g\in D}w\br{f}\cdot w\br{g}\cdot q_A(f,g).
    \end{equation}
    We claim that each pair $f,g\in D$ satisfies
    \begin{equation}\label{eq:tree-demand-pair-cover}
      q_{S^*}(f,g)
        \leq\sum_{e\in E^*}q_{S_e}(f,g).
    \end{equation}
    Since $q_{S^*}(f,g)\in \brc{0,1,2}$ we distinguish two non-trivial cases:
    \begin{enumerate}
      \item $q_{S^*}(f,g)=1$. After exchanging $f$ and $g$ if necessary, there is an endpoint $u$ of $f$ in $S^*$ and an endpoint $v$ of $g$ in $\bar{S}^*$. The unique $u$-$v$ path in $T$ contains an edge $e\in E^*$. Removing $e$ places $u$ and $v$ in opposite components. Therefore, either $f\in\delta_D[S_e]$ and $g\in\delta_D[\bar S_e]$, or $g\in\delta_D[S_e]$ and $f\in\delta_D[\bar S_e]$. In either case $q_{S_e}(f,g)\geq 1$.
      \item $q_{S^*}(f,g)=2$. Then each of $f$ and $g$ has one endpoint in $S^*$ and one in $\bar{S}^*$. Let $P_f$ and $P_g$ be the paths in $T$ joining the endpoints of $f$ and $g$, respectively. Choose $f'\in P_f\cap E^*$ and $g'\in P_g\cap E^*$.
    If $f'\neq g'$, then $f$ is incident to both components of $T\setminus f'$, so $q_{X_{f'}}(f,g)\geq1$ regardless of which component contains an endpoint of $g$. The symmetric argument gives $q_{X_{g'}}(f,g)\geq1$, and consequently $q_{X_{f'}}(f,g)+q_{X_{g'}}(f,g)\geq2$. If $f'=g'$, both $f$ and $g$ are incident to both components of $T\setminus f'$, and hence $q_{X_{f'}}(f,g)=2$.
    \end{enumerate}
    This proves~\eqref{eq:tree-demand-pair-cover}.
    Combining this with \eqref{eq:tree-symmetrized-demand-pairs} applied first to $S^*$ and then to each $S_e$ gives
    \begin{align*}
      W^* &= \frac{1}{2}\cdot \sum_{f,g\in D}w\br{f}\cdot w\br{g}\cdot q_{S^*}(f,g) \leq
      \frac{1}{2}\cdot \sum_{f,g\in D}w\br{f}\cdot w\br{g}\cdot \sum_{e\in E^*}q_{S_e}(f,g)
       \\&= \sum_{e\in E^*} \frac{1}{2}\cdot \sum_{f,g\in D}w\br{f}\cdot w\br{g}\cdot q_{S_e}(f,g)
       = \sum_{e\in E^*} W_e,
    \end{align*}
    as required.
  \end{proof}

  Let $E^*_+=\brc{e\in E^*:W_e>0}$. Lemma~\ref{lemma:tree_edge_weight_sum} and $W^*>0$ imply that $E^*_+\neq\emptyset$. Thus
  \[
    \psi_{w,T}\br{S^*}
    =\frac{\sum_{e\in E^*}c\br{e}}{W^*}
    \geq\frac{\sum_{e\in E^*}c\br{e}}{\sum_{e\in E^*}W_e}
    \geq\frac{\sum_{e\in E^*_+}c\br{e}}{\sum_{e\in E^*_+}W_e}
    \geq\min_{e\in E^*_+}\brc{\frac{c\br{e}}{W_e}}
    =\min_{e\in E^*_+}\brc{\psi_{w,T}\br{S_e}},
  \]
  where the first inequality uses Lemma~\ref{lemma:tree_edge_weight_sum}. The second discards the edges with $W_e=0$, while the last inequality is due to averaging.
\end{proof}

\subsection{General graphs}

To tackle general supply graphs, we transform demands into weighted degrees, allowing us to approximate generalized conductance using the well-known \textsc{Min-Ratio Vertex Cut} problem. Its input is a graph $H$ with two non-negative vertex functions $\pi_1$ and $\pi_2$. A feasible solution is a partition $(A,B,Z)$ of $V_H$ with no edges between $A$ and $B$, and its \emph{value} is
\[
  \alpha_{\pi_1,\pi_2}\br{A,B,Z}
  =\frac{\pi_1\br{Z}}
  {\pi_2\br{A\cup Z}\cdot\pi_2\br{B\cup Z}}.
\]

Now, define the \emph{degree conductance} of a cut $S\subseteq V$ as
\[
  \chi_{w}\br{S}
  =\frac{c\br{S,\bar{S}}}{\vol_D\br{S}\cdot \vol_D\br{\bar{S}}}.
\]
For every cut of finite value, Equation~\eqref{eq:demand-degree-identity} gives $
  w\br{S,V}\leq \vol_D\br{S}\leq 2w\br{S,V}$, so
\begin{equation} \label{eq:factor4}
  \frac{\psi_{w}\br{S}}{4}\leq
  \frac{c\br{S,\bar{S}}}{\vol_D\br{S}\cdot \vol_D\br{\bar{S}}}=\chi_{w}\br{S}\leq
  \psi_{w}\br{S}.
\end{equation}

\begin{theorem}\label{thm:generalized-conductance-approx}
  There is a polynomial-time $\bigo\br{\sqrt{\log n}}$-approximation algorithm for \textsc{Generalized Conductance} on general graphs.
\end{theorem}
\begin{proof}
  Given $G=(V,E,D,c,w)$, subdivide every supply edge $e=uv\in E$ with a new vertex $x_e$, forming
  \[
    \begin{gathered}
      V_H=V\cup\brc{x_e\colon e\in E}
      \qquad\text{and}\qquad
      E_H=\brc{ux_e,x_ev\colon e=uv\in E},\\
      \pi_1\br{v}=\infty,\quad \pi_2\br{v}=d_D\br{v}\quad(v\in V),\quad
      \pi_1\br{x_e}=c\br{e},\quad \pi_2\br{x_e}=0\quad(e\in E).
    \end{gathered}
  \]

  It is easy to see that this construction encodes degree conductance exactly: every cut $S\subsetneq V$ in $G$ corresponds to a vertex cut of value $\chi_{w}\br{S}$ in $H$, and every finite-value vertex cut $(A,B,Z)$ projects to $S=A\cap V$ with $\chi_{w}\br{S}\leq\alpha_{\pi_1,\pi_2}\br{A,B,Z}$.
  The $\bigo\br{\sqrt{\log |V_H|}}$-approximation algorithm for \textsc{Min-Ratio Vertex Cut} of \cite{ImprovedApproximationAlgorithmsForMinimumWeightVertexSeparators} (with projected output to a cut $S$ of $G$), Equation~\eqref{eq:factor4}, and observing that $|V_H|=n+|E|=\bigo\br{n^2}$ imply an $\bigo\br{\sqrt{\log n}}$-approximation algorithm for \textsc{Generalized Conductance}.
\end{proof}

\section{Applications of the Generalized Conductance}
This section provides our main algorithmic result for GPD and HCD on general graphs.


\subsection{Graph Partitioning with Demands}

We apply the approximation algorithm for GC to obtain an approximation algorithm for \textsc{Graph Partitioning with Demands} (GPD). In the latter problem, we are given an instance $G=(V,E,D,c,w)$, together with a fixed constant $\rho\in(0,1)$. The goal is to find a subset of supply edges $C\subseteq E$ that minimizes $c\br{C}$ such that for every connected component $H$ of $G\setminus C$, we have $w\br{V_H}\leq B_\rho$.
Let $\OPTdem{G}{\rho}=c\br{C^*}$, where $C^*\subseteq E$ is an optimal solution.
Let $\rho<\eta<1$ be a constant. We say that $C\subseteq E$ is $\eta$-feasible if for every connected component $H$ of $G\setminus C$, it holds $w\br{V_H}\leq B_\eta$, instead of $w\br{V_H}\leq B_\rho$.
We will obtain a bi-criteria approximation algorithm for GPD that returns an $\eta$-feasible subset of edges $C$ such that $c\br{C}=\bigo_{\rho, \eta}\br{\gamma\br{n}}\cdot \OPTdem{G}{\rho}$, where $n=|V|$ and $\gamma\br{n}$ is an upper bound on the approximation ratio of the chosen \textsc{Generalized Conductance} algorithm on instances with at most $n$ vertices. This is done in two parts. Firstly, we show an algorithm for a limited range of $\eta$ and then extend it to the entirety of $\eta\in(\rho,1)$.

\subsubsection{Bi-criteria approximation algorithm for a limited range of $\eta$}
Let $
\delta_{\rho}=\min\brc{1/4,1-\rho}$ and $
\eta\in\br{\max\brc{1/2,1-\delta_{\rho}},1}$.
The analysis uses a generalization of the standard technique of dealing with \textsc{Balanced Cut}/\textsc{Graph Partitioning} problems using \textsc{Sparsest Cut} \cite{ImprovedApproximationAlgorithmsForMinimumWeightVertexSeparators,MulticommodityMaxFlowMinCutTheoremsAndTheirUseInDesigningApproximationAlgorithms}. However, due to the fact that we need to control both the demand in the components and the demand crossing the cut, a more careful consideration is required. We use a simple algorithm, characterized by the following claim.

\begin{theorem}\label{thm:gpd-bicriteria}
    For every fixed $\rho\in(0,1)$, let $\delta_{\rho}=\min\brc{1/4,1-\rho}$ and let $\gamma\br{n}$ be the approximation ratio of GC. For every fixed $\eta\in\br{\max\brc{1/2,1-\delta_{\rho}},1}$, there exists a polynomial-time algorithm that returns an $\eta$-feasible subset of edges $C$ such that
    \[
    c\br{C}=\frac{\gamma\br{n}}{(\eta+\delta_{\rho}-1)^2}\cdot \OPTdem{G}{\rho}.
    \]
\end{theorem}
\begin{proof}
    The algorithm is as follows:

    \begin{algorithm}[H]
    \label{algo:partitioning-limited-range}
    \caption{Limited-range bi-criteria partitioning}
    \textbf{Input} Instance $G=(V,E,D,c,w)$, parameters $\rho,\eta$, and a $\gamma\br{n}$-approximation oracle for \textsc{Generalized Conductance}, where $n=|V|$\;
    \textbf{Output} An $\eta$-feasible supply-edge set $C\subseteq E$\;
    Set $B_\eta\gets\eta W$\;
    $C\gets\emptyset$. $G'\gets G$\;
    \While{$w_{G'}\br{V_{G'}}>B_\eta$}{
        Find a cut $(S,\bar{S})$ in $G'$ using the $\gamma\br{n}$-approximation for GC, and assume that $w_{G'}\br{S}\leq w_{G'}\br{\bar{S}}$\;
        $C\gets C\cup E_{G'}\br{S,\bar{S}}$\;
        Remove $S$ from $G'$\;
    }
    \textbf{return} $C$\;
\end{algorithm}

    \begin{lemma}
    Every component $H\in G\setminus C$ is $\eta$-feasible.
    \end{lemma}
    \begin{proof}
    If $H$ is the same graph as $G'$ after the last iteration of the while loop, then by the halting condition of the loop, $w\br{V_H}\leq B_\eta$. Otherwise, $H$ was the smaller side of the cut in some iteration of the while loop, and thus $w\br{V_H}\leq W/2\leq B_\eta$, where we use the assumption $\eta\in\br{\max\brc{1/2,1-\delta_{\rho}},1}$ (in particular, $\eta\geq 1/2$). In either case $H$ is $\eta$-feasible.
    \end{proof}

    Consider an optimal solution $C^*\subseteq E$, $\OPTdem{G}{\rho}=c\br{C^*}$.
    We want to reason about $c\br{C}$.
    To do so, we consider the connected components of $G\setminus C^*$ and partition them into two vertex sets $A$ and $B$, each containing a significant fraction of the total demand as stated in the next lemma.
    \begin{lemma} \label{lem:ABpartition}
                For every instance $G=(V,E,D,c,w)$, there exists a partition of the connected components of $G\setminus C^*$ into two vertex sets $A$ and $B$ such that $w\br{A, V}\geq \delta_{\rho}\cdot W$ and $w\br{B, V}\geq \delta_{\rho}\cdot W$.
    \end{lemma}
    \begin{proof}
    We have the following cases:
    \begin{enumerate}
                \item There exists a connected component $H$ of $G\setminus C^*$ such that $w\br{V_H}\geq \delta_{\rho}\cdot W$.
                Set $A=V_H$ and $B=V\setminus V_H$. We have
                $w\br{A, V}\geq w\br{A} \geq \delta_{\rho}\cdot W$
                and 
                $w\br{B, V}=W-w\br{A}\geq \br{1-\rho}\cdot W\geq \delta_{\rho}\cdot W$,
                since $w\br{A}=w\br{V_H}\leq B_\rho$. The last inequality uses the definition $\delta_{\rho}=\min\brc{1/4,1-\rho}$.
        \item For every connected component $H$ of $G\setminus C^*$, we have $w\br{V_H}\leq \delta_{\rho}\cdot W$. We proceed in iterations. Start with $A=\emptyset$ and keep adding vertices of the subsequent components of $G\setminus C^*$ to $A$ until $w\br{A}\geq \delta_{\rho}\cdot W$. Let $H$ be the last component added to $A$ and let $A'$ be $A$ one iteration before that, i.e., $A'=A\setminus V_H$.
        Consider the two cases:
        \begin{itemize}
                        \item If $w\br{A', V_H}\geq \delta_{\rho}\cdot W$, then take a step back: set $A=A'$ and $B=V\setminus A'$. We have
            \[
                        w\br{A, V}\geq w\br{A, B} \geq w\br{A', V_H}\geq \delta_{\rho}\cdot W
            \]
                        and similarly, using $w\br{A,B}=w\br{B,A}$ for disjoint $A$ and $B$, $w\br{B, V}\geq \delta_{\rho}\cdot W$.
            \item Otherwise, $A$ stays unchanged and we set $B=V\setminus A$. We have $w\br{A, V}\geq \delta_{\rho}\cdot W$ by the halting condition of the above greedy rule, and
              \[
              w\br{B, V}
              = 
              W-w\br{A}= W-\br{w\br{A'}+w\br{A', V_H}+w\br{V_H}}
              \geq 
                            W-3\delta_{\rho}\cdot W,
              \]
            This gives $w\br{B, V}\geq \delta_{\rho}\cdot W$, because $\delta_{\rho}\leq 1/4$, again by $\delta_{\rho}=\min\brc{1/4,1-\rho}$.
        \end{itemize}
    \end{enumerate}
    This concludes the proof of the lemma.
  \end{proof}
    Take the partition $\br{A,B}$ from Lemma~\ref{lem:ABpartition}.
        Suppose that we are at some iteration when $w\br{V_{G'}}> B_\eta$. We have that
    \[
        w\br{A\cap V_{G'}, V_{G'}}\geq w\br{A, V}-w\br{V\setminus V_{G'}, V}> \delta_{\rho}\cdot W-\br{W-B_\eta}=\br{\eta+\delta_{\rho}-1}\cdot W.
    \]
        Here the halting condition $w\br{V_{G'}}> B_\eta$ implies $w\br{V\setminus V_{G'}, V}=W-w\br{V_{G'}}<W-B_\eta$.
        Symmetrically, for $B$ we get that $w\br{B\cap V_{G'}, V_{G'}}\geq \br{\eta+\delta_{\rho}-1}\cdot W$. Since $\eta>1-\delta_{\rho}$ by assumption on $\eta$, we have $\eta+\delta_{\rho}-1>0$. Let $\br{S, \bar{S}}$ be the cut returned by the $\gamma\br{n}$-approximation algorithm for the \textsc{Generalized Conductance} in $G'$ in the considered iteration. Since $\br{A\cap V_{G'}, B\cap V_{G'}}$ is a cut of $G'$, we have
\begin{align*}
        \frac{c\br{S, \bar{S}}}{w\br{S, V_{G'}}\cdot w\br{\bar{S}, V_{G'}}}\leq \frac{\gamma\br{n} \cdot c\br{A, B}}{w\br{A\cap V_{G'}, V_{G'}}\cdot w\br{B\cap V_{G'}, V_{G'}}}
        \leq \frac{\gamma\br{n}\cdot c\br{A, B}}{\br{\eta+\delta_{\rho}-1}^2\cdot W^2}.
\end{align*}
Now, by rearranging the above inequality we get:
\begin{align*}
        c\br{S, \bar{S}} \leq 
     \frac{\gamma\br{n}\cdot c\br{A, B}}{\br{\eta+\delta_{\rho}-1}^2\cdot W^2}\cdot w\br{S, V_{G'}}\cdot w\br{\bar{S}, V_{G'}}
     \leq \frac{\gamma\br{n}\cdot \OPTdem{G}{\rho}\cdot w\br{S, V_{G'}}}{\br{\eta+\delta_{\rho}-1}^2\cdot W}
\end{align*}
since trivially we have $c\br{A, B}\leq \OPTdem{G}{\rho}$ and $w\br{\bar{S}, V_{G'}}\leq w\br{V_{G'}}\leq W$.

Let $\ell$ be the index of some iteration, let $G_{\ell}'$ be $G'$ at this iteration and let $\br{S_{\ell}, \bar{S_{\ell}}}$ be the cut at this iteration. By summing the above inequality over all iterations we get:
\begin{align*}
    c\br{C} &\leq \sum_{\ell} c\br{S_{\ell}, \bar{S_{\ell}}} \leq \frac{\gamma\br{n}\cdot \OPTdem{G}{\rho}}{\br{\eta+\delta_{\rho}-1}^2\cdot W}\cdot \sum_{\ell} w\br{S_{\ell}, V_{G'_{\ell}}} \leq \frac{\gamma\br{n}}{\br{\eta+\delta_{\rho}-1}^2}\cdot \OPTdem{G}{\rho}
\end{align*}
where the last inequality follows since the sets $S_{\ell}$ are disjoint, so $\sum_{\ell} w\br{S_{\ell}, V_{G'_{\ell}}}\leq W$.
\end{proof}

As an immediate corollary (by Theorems~\ref{thm:tree-generalized-conductance} and~\ref{thm:generalized-conductance-approx}) we get:
\begin{corollary}
    For fixed $\rho\in(0,1)$ and $\eta\in\br{\max\brc{1/2,1-\delta_{\rho}},1}$ (where $\delta_{\rho}=\min\brc{1/4,1-\rho}$), there is a bi-criteria approximation algorithm for \textsc{Graph Partitioning with Demands} which returns an $\eta$-feasible cut $C$ with the following approximation ratios:
    \[
    \bigo\br{\frac{\sqrt{\log n}}{\br{\eta+\delta_{\rho}-1}^2}}\cdot \OPTdem{G}{\rho}
    \quad\text{on general graphs,}\quad \frac{1}{\br{\eta+\delta_{\rho}-1}^2}\cdot \OPTdem{G}{\rho}
    \quad\text{on trees.}
    \]
\end{corollary}

\subsubsection{Extending the range of the parameters}
Next, we show how to extend the above procedure into a recursive algorithm that works for any choice of $0<\rho < \eta<1$. 
Since for $\rho\geq 3/4$ the previous procedure works for any $\eta\in(\rho,1)$, we focus on the case $\rho<3/4$. Then observe that $\delta_{\rho}=1/4$, and therefore we know that the previous procedure works for any $\eta\in[3/4,1)$.
The next result shows an $\bigo_{\rho,\eta}\br{\gamma\br{n}}$-approximation for any choice of $0<\rho<\eta<1$. Note that the constants hidden in the $\bigo_{\rho,\eta}\br{\cdot}$ are quite large and we do not attempt to optimize them.

\begin{theorem}\label{thm:gpd-extension}
    Fix $0<\rho<3/4\leq \eta< 1$, and let $
        \kappa\br{\rho,\eta}=\max\brc{400,\ \frac{1}{\br{\br{1-\rho/\eta}\cdot 4/5}^2}}.$
    There exists a polynomial-time algorithm that outputs an $\eta$-feasible subset of edges $C_G$ such that $c\br{C_G}\leq d\cdot \kappa\br{\rho,\eta}\cdot \gamma\br{n}\cdot \OPTdem{G}{\rho},
    $
    where $d=\cl{\log_{5/4}\br{1/\eta}}$ and $\gamma\br{n}$ is the approximation factor of a polynomial-time algorithm for \textsc{Generalized Conductance}.
\end{theorem}
\begin{proof}
The recursive procedure for a component $H$ is as follows.

\begin{algorithm}[H]
    \caption{Recursive extension to the full parameter range}\label{alg:recursive-partitioning}
    \textbf{Input} A component $H$, global budgets $B_\rho=\rho\cdot W$ and $B_\eta=\eta\cdot W$, and the algorithm from Theorem~\ref{thm:gpd-bicriteria} (Algorithm~\ref{algo:partitioning-limited-range})\;
    \textbf{Output} An $\eta$-feasible supply-edge set $C_H$\;
    \If{$w\br{V_H}\leq B_\eta$}{
        \textbf{return} $\emptyset$\;
    }
    $\eta_0\gets\max\brc{4/5,\,B_\eta/w\br{V_H}}$\;
    Run Algorithm~\ref{algo:partitioning-limited-range} on $H$ with parameter $\eta_0$. Let $C$ be the returned supply-edge set\;
    $C_H\gets C$\;
    \ForEach{connected component $H'$ of $H\setminus C$}{
        Recursively partition $H'$ and add the returned supply-edge set $C_{H'}$ to $C_H$\;
    }
    \textbf{return} $C_H$\;
\end{algorithm}

The proof of correctness is deferred to Appendix~\ref{app:gpd-extension}.
\end{proof}

As an immediate corollary we get:
\begin{corollary}\label{cor:graph-partitioning-approximation}
    Fix $0<\rho<\eta<1$. There is a polynomial-time approximation algorithm for \textsc{Graph Partitioning with Demands} that returns an $\eta$-feasible subset of edges $C_G$ and has approximation guarantees:
    \[
    \bigo_{\rho,\eta}\br{\sqrt{\log n}}\cdot \OPTdem{G}{\rho}
    \quad\text{on general graphs,}\quad
    \bigo_{\rho,\eta}\br{1}\cdot \OPTdem{G}{\rho}
    \quad\text{on trees.}
    \]
\end{corollary}

\subsection{Hierarchical clustering with demands}

We now concern ourselves with \textsc{Hierarchical Clustering with Demands} (HCD) problem. The input is a graph $G=(V,E,D,c,w)$. A \emph{clustering tree} $\cT$ is a rooted tree whose vertices are subgraphs of $G$ called \emph{clusters}, representing a laminar family of subgraphs of $G$, i.e., for any two clusters $H$ and $H'$, either $V_H\subseteq V_{H'}$, $V_{H'}\subseteq V_H$, or $V_H\cap V_{H'}=\emptyset$. We require the clustering tree to have the entire graph $G$ as its root cluster, any non-leaf cluster to have at least two child clusters whose vertex sets form a partition of the vertex set of the parent cluster, and each leaf to be a singleton, i.e., a cluster with a single vertex of $G$. Therefore, there is a one-to-one correspondence between the leaves of $\cT$ and the vertices of $G$.
For $u,v\in V$, let $\operatorname{lca}_{\cT}(u,v)$ denote their lowest common ancestor in $\cT$, that is, the cluster that is farthest from the root and contains both $u$ and $v$. For a non-leaf node with cluster $H$, let $C_H$ be the supply edges whose endpoints lie in distinct child clusters of $H$ (for simplicity we equate the nodes with the clusters). The objective is to minimize the total cost of separating the vertices according to the demand function, which is captured by the following function:
\begin{equation}\label{eq:Dasgupta}
    c\br{\cT}
    =\sum_{H\in\cT}
        w\br{V_H}\cdot c\br{C_H}
    =\sum_{uv\in E}c\br{uv}\cdot
        w\br{V_{\operatorname{lca}_{\cT}(u,v)}}.
\end{equation}

We parameterize our top-down algorithm by constants $0<\rho<\eta<1$. Assume that we have a polynomial-time $\gamma_{\rho,\eta}\br{n}$-approximate subroutine for \textsc{Graph Partitioning with Demands}, denoted by $\DPartition$, suppressing its parameters $\rho$ and $\eta$.

\begin{theorem}\label{thm:hc-gpd-transfer}
    Fix constants $0<\rho<\eta<1$ and suppose that an $\gamma_{\rho,\eta}(n)$-approximate subroutine for GPD, called $\DPartition$ is given. Then, there exists a $\frac{\lceil1/\rho\rceil\cdot\gamma_{\rho,\eta}(n)}{1-\eta}$-approximation polynomial-time algorithm for \textsc{Hierarchical Clustering with Demands}.
\end{theorem}
\begin{proof}
    The following algorithm recursively partitions each non-singleton cluster using the subroutine for GPD and constructs the clustering tree $\cT$ in a top-down manner.
    
\medskip
\begin{algorithm}[H]
    \caption{Top-down hierarchical clustering with demands}
    \label{algo:hclustering}
    \textbf{Input} Graph $G$, parameters $0<\rho<\eta<1$, and the subroutine $\DPartition$\;
    \textbf{Output} A clustering tree $\cT$\;
    Initialize $\cT$ with root cluster $G$\;
    \textbf{Procedure} $\textsc{BuildHierarchy}(H)$\;
    \Indp
            \If{$w\br{V_H}=0$}{
                Make the vertices of $H$ singleton children\;
            }
            \Else{
                Compute $C_H$ by applying $\DPartition$  on $H$\;

                Make the components of $H\setminus C_H$ the children of $H$\;

            \ForEach{child $K$ of $H$}{
                $\textsc{BuildHierarchy}(K)$\;
            }                }
    \Indm
    $\textsc{BuildHierarchy}(G)$\;
    \textbf{return} $\cT$\;
\end{algorithm}

\medskip

    Let $\cT^*$ be an optimal clustering tree, i.e., one that minimizes~\eqref{eq:Dasgupta}. For $q\geq0$, define
    \[
        E_q^*
        =\left\{
            uv\in E\colon
            w\br{V_{\operatorname{lca}_{\cT^*}(u,v)}}>q
        \right\}.
    \]
    By summing up the contribution of every supply-edge over the levels below the demand of its least common ancestor, we get that $
        c\br{\cT^*}
        =\sum_{q=0}^{W}c\br{E_q^*}.
    $

    Fix a non-singleton cluster $H$ produced by the algorithm and 
    write $r_H=w\br{V_H}$. The maximal clusters of $\cT^*$ having demand at most $\fl{\rho r_H}$ partition $V$, and by extension $V_H$, into clusters of demand at most $\fl{\rho r_H}$.
    Consequently, the set $E_{\fl{\rho r_H}}^*\cap E_H$ is a feasible solution to \textsc{Graph Partitioning with Demands} on the graph $H$ with target $\rho$, so $
        c\br{C_H}\leq \gamma_{\rho,\eta}\br{\spr{V_H}} \cdot c\br{E_{\fl{\rho r_H}}^*\cap E_H}.$
    Let $s_H$ be the largest demand among the children of $H$. Since the subroutine returns an $\eta$-feasible solution, $s_H\leq\eta \cdot r_H$. Hence
    \[
        r_H\cdot c\br{C_H}
        \leq\frac{\gamma_{\rho,\eta}\br{\spr{V_H}}}{1-\eta}
            \cdot\br{r_H-s_H}\cdot c\br{E_{\fl{\rho r_H}}^*\cap E_H} \leq\frac{\gamma_{\rho,\eta}\br{\spr{V_H}}}{1-\eta}\cdot
            \sum_{t=s_H+1}^{r_H}c\br{E_{\fl{\rho t}}^*\cap E_H}
    \]
    where the first inequality is by $r_H-s_H\geq(1-\eta)\cdot r_H$. Since for every integer $t\leq r_H$, $E_{\fl{\rho r_H}}^*\subseteq E_{\fl{\rho t}}^*$, summing this inclusion over $t=s_H+1,\ldots,r_H$ gives the second inequality.

    For a fixed $t$, the clusters satisfying $s_H<t\leq r_H$ have pairwise disjoint induced edge sets. Indeed, if one such cluster were a descendant of another, it would lie in a non-singleton child of the ancestor and would have demand at most the ancestor's value $s_H<t$, contradicting $t\leq r_H$ for the descendant. Summing over all clusters, we conclude that
    \begin{align*}
        c\br{\cT}
        \leq\frac{\gamma_{\rho,\eta}(n)}{1-\eta}\cdot
            \sum_{t=0}^{W}c\br{E_{\fl{\rho t}}^*}
            \leq\frac{\lceil1/\rho\rceil\cdot \gamma_{\rho,\eta}(n)}{1-\eta}\cdot
            \sum_{t=0}^{W}c\br{E_{t}^*}
        =\frac{\lceil1/\rho\rceil\cdot\gamma_{\rho,\eta}(n)}{1-\eta}\cdot
            c\br{\cT^*},
    \end{align*}
    where the first inequality uses the disjointness of clusters on each level. For each $q$, at most $\lceil1/\rho\rceil$ integers $t$ satisfy $\fl{\rho t}=q$, so reindexing gives the second inequality.
\end{proof}

For the approximation bounds below, we apply Theorem~\ref{thm:hc-gpd-transfer} with $\rho=1/2$.

\begin{corollary}\label{cor:hc-general-bounds}
    There is an $\bigo\br{\sqrt{\log n}}$-approximation algorithm for \textsc{Hierarchical Clustering with Demands} on general graphs and an $\bigo\br{1}$-approximation algorithm on trees.
\end{corollary}

\section{Graph Partitioning with Demands on Trees}\label{sec:tree-configuration-lp}

We now focus on the case where the supply graph is a tree.
Since Section~\ref{sec:gpd-hardness} provides evidence that the problem cannot be approximated efficiently even on supply stars, we focus on improving the approximation guarantee given by Theorem \ref{thm:gpd-extension}, since the constant provided by that theorem is relatively large.
We give two polynomial-time bi-criteria algorithms: anchored-LP rounding (cf. Section~\ref{sec:Anchor-LP-based-algorithm}) and separator recursion based on the \textsc{Maximum Coverage} problem (cf. Section~\ref{sec:Separator-Recursion-via-Maximum-Coverage}). We also give a single-criteria separator algorithm based on \textsc{Partial Multicut} with $n^{O\br{1/\rho}}$ runtime (cf. Section~\ref{sec:Single-Criteria-Approximation-via-Multicut}).
Taking the best of these guarantees gives approximation ratios depending on $\rho$ and $\eta$ that improve on the constant provided by Theorem~\ref{thm:gpd-extension}.
As a corollary, we also obtain a polynomial-time $\frac{4e}{e-1}$-approximation for \textsc{Hierarchical Clustering with Demands} on supply trees. 

\subsection{Anchor-LP-based algorithm}
\label{sec:Anchor-LP-based-algorithm}

Our first algorithm is based on an anchored linear programming relaxation. For each vertex $a\in V$, which we call an \emph{anchor}, we need a constraint upper bounding the internal demand of the component containing $a$. For this reason our LP consists of two kinds of variables, the $x_e$ variables, corresponding to the edges that are cut and the $z^a_{uv}$ variables, denoting whether a given demand edge is internal to the component containing $a$.
After solving the anchored LP, let $x$ be its optimal edge-length vector. We pick an arbitrary global root $r\in V$ and calculate $\dist(u)=\dist_x(r,u)$ for every vertex $u\in V$. This projects the vertices of the tree onto the real line with $r$ at the origin and $u$ at $\dist_x(r,u)$. Then, we choose an appropriate value of $R$ and subdivide the line into intervals of length $R$, and cut every edge that crosses an interval boundary. To make this approach work, we draw a random phase $\theta\in[0,R)$ and shift all intervals by $\theta$, which gives a solution with low expected cost. 


For $a,u\in V$, let $P_T(a,u)$ be the unique path from $a$ to $u$ in the input tree $T=(V,E,D,c,w)$. Introduce variables $x_e\in[0,1]$ for each edge $e\in E$ and define $
    \dist_x(a,u)=\sum_{e\in P_T(a,u)}x_e.$
For every vertex $a\in V$, called an \emph{anchor}, and every demand edge $uv\in D$, introduce a variable $z^a_{uv}\in[0,1]$. The \emph{anchored relaxation} of GPD is
\begin{equation}\label{lp:tree-anchored}
\begin{aligned}
    \min\quad
        &\sum_{e\in E}c(e)\cdot x_e\\
    \text{subject to}\quad
        &z^a_{uv}\geq1-\dist_x(a,u)-\dist_x(a,v)
            &&\forall a\in V,\ uv\in D,\\
        &\sum_{uv\in D}w(u,v)\cdot z^a_{uv}\leq B_\rho
            &&\forall a\in V,\\
        &x, z \in [0,1]^{E \cup (V \times D)}
\end{aligned}
\end{equation}
It is easy to see that this is a valid relaxation of the original problem.
\begin{theorem}\label{thm:tree-anchored-rounding}
    For every fixed $0<\rho<\eta<1$ and arbitrary pair demands on a supply tree, there is a polynomial-time algorithm that returns an $\eta$-feasible supply-edge set $C\subseteq E$ satisfying $
        c(C)\leq\frac{2\eta}{\eta-\rho}\cdot \OPTdem{T}{\rho}.$
\end{theorem}
\begin{proof}
    Solve the above LP, let $(x,z)$ be an optimal solution and choose $
        R=\frac{\eta-\rho}{2\eta}.$
    Then $0<R<1/2$, $B_\rho/(1-2R)=B_\eta$, and $1/R=2\eta/(\eta-\rho)$.
    Root $T$ at an arbitrary vertex $r$ and for each $v\in V$ define $
        \dist(v)=\dist_x(r,v).
    $
    The values $\dist(v)$ embed the tree vertices into the non-negative real line. 
    We now wish to choose thresholds along this line to determine which edges to cut. Let $F=\emptyset$.
    Choose a \emph{phase} $\theta$ uniformly at random from $[0,R)$. Include a parent-child edge $e=pv$ in $F$ whenever $\dist(p)<\theta+k\cdot R\leq \dist(v)$
    for some integer $k$. If $x_e<R$, the set of phases $\theta\in[0,R)$ that place a threshold in $(\dist(p),\dist(v)]$ has measure $x_e$, so $e$ is included in $F$ with probability $x_e/R$. If $x_e\geq R$, every phase places at least one threshold in this interval and the inclusion probability is $1$. Thus in all cases $
        \Pr[e\in F]=\min\brc{1,\frac{x_e}{R}}\leq\frac{x_e}{R},$
    and therefore $
        \mathbb{E}[c(F)]\leq\frac{1}{R}\cdot \sum_{e\in E}c(e)\cdot x_e.
    $

    Now fix a component $H$ of $T\setminus F$, and let $a$ be its vertex closest to the root. For every $v\in V_H$, the path from $a$ to $v$ lies entirely in $H$. No threshold of the form $\theta+kR$ can lie in $(\dist(a),\dist(v)]$, since such a threshold would cause an edge of that path to be deleted. The thresholds are spaced $R$ apart, so $
        \dist_x(a,v)=\dist(v)-\dist(a)<R.
    $
    For each demand edge $uv\in D_H$, the LP connectivity constraint for anchor $a$ gives
    \[
        z^a_{uv}\geq1-\dist_x(a,u)-\dist_x(a,v)>1-2R.
    \]
    The anchor-demand constraint and non-negativity of $z$ imply
    \[
        B_\rho\geq\sum_{uv\in D}w(u,v)\cdot z^a_{uv}
          \geq\sum_{uv\in D_H}w(uv)\cdot z^a_{uv}
          \geq(1-2R)\cdot w(V_H).
    \]

    It remains to derandomize the procedure. Since $F$ changes only when $\theta$ crosses one of the at most $|V|$ residues $\dist(v)\bmod R$, check one value per each interval between residues and pick the candidate $C$ with minimum cost. By the properties of the expected value
    \[
        c(C)\leq \mathbb{E}[c(F)]\leq \frac{1}{R}\cdot \sum_{e\in E}c(e)\cdot x_e
        \leq\frac{2\eta}{\eta-\rho}\cdot \OPTdem{T}{\rho}.
    \]
    and every component $H$ of $T\setminus C$ satisfies $w(V_H)\leq B_\rho/(1-2R)=B_\eta$.
\end{proof}

\subsection{Separator Recursion via Maximum Coverage}
\label{sec:Separator-Recursion-via-Maximum-Coverage}

We now show a combinatorial algorithm for \textsc{Graph Partitioning with Demands} on trees. The idea is to find a small set of vertices whose removal leads to components with sufficiently low demand. Then, the new goal of the algorithm becomes to remove a large enough amount of demand from each vertex of the separator. To this end, we employ a well-known subproblem called \textsc{Maximum Coverage}.  We start with the following lemma

\begin{lemma}\label{lem:tree-separator}
    Fix an integer $t\geq2$ and take any subtree $H$. There is a polynomial-time procedure $\textsc{Separator}(H, w, t)$ that returns a vertex separator $\mathcal L_H\subseteq V_H$ with $|\mathcal L_H|\leq t-1$ such that every component $J$ of $H\setminus\mathcal L_H$ satisfies $
        w(V_J)\leq\frac{w(V_H)}{t}.$
\end{lemma}
\begin{proof}
The proof is relatively easy and is deferred to Appendix~\ref{app:tree-separator}.
\end{proof}


In our algorithm we apply Lemma~\ref{lem:tree-separator} with the smallest integer $t\geq e$, which is $t=3$. 
Thus, $\mathcal L_H$ contains at most two separators, $\spr{\mathcal{L_H}}\leq 2$, and every component of $H\setminus\mathcal L_H$ has demand at most $w(V_H)/3$.
Let $S_H=H[\mathcal L_H]$, (note that here $S_H$ is either a vertex or a path). 
We enumerate $Q=\emptyset$ and $Q=\{e\}$ for each edge $e\in E_{S_H}$. Thus at most $1+|E_{S_H}|\leq n$ guesses are considered. For a guess $Q$, let $\mathcal K_Q$ be the components of $H\setminus Q$.
Define $S_K=K[V_{S_H}]$.

For $A\subseteq E_K\setminus E_{S_K}$, let $U_K(A)$ be the set of vertices separated from $S_K$ after cutting the edges in $A$, and let
$
    f_K(A)=w\br{U_K(A), V_K}
$
be the demand separated from $S_K$ by cutting $A$.
Observe that $f_K$ is a weighted coverage function on the demand pairs. 
For $w\br{V_K}>B_\rho$, set $q_K=w\br{V_K}-B_\rho$ and define
$
    \OPT_{q_K}\br{K}
    =\min\brc{c_K(A)\colon A\subseteq E_K\setminus E_{S_K},f_K(A)\geq q_K}
$
if there exists any $A\subseteq E_K\setminus E_{S_K}$ with $f_K(A)\geq q_K$.  Now we claim that the $ (1-1/e)$-approximation for \textsc{Maximum Coverage} of~\cite{KhullerMossNaor1999BudgetedMaximumCoverage} provides us with a polynomial-time subroutine $\textsc{CoreQuotaCover}(K,S_K,f_K,c_K,q_K)$ returning $A_K$ such that
\[
    f_K(A_K)\geq(1-1/e)\cdot q_K,\quad \quad\text{and}\quad
    c_K(A_K)\leq\OPT_{q_K}\br{K}.
\]
Indeed, binary-search the budget $b\in \brc{0,\dots, c_K(E_K)}$ and run the algorithm for \textsc{Maximum Coverage} at each trial.
Let $b$ be the lowest budget such that $f_K(A_K)\geq(1-1/e)\cdot q_K$, so that we get $b\leq \OPT_{q_K}\br{K}$.
 Since we use binary search the required number of trials is $O\br{\log\br{c_K(E_K)}}$, which is polynomial in the input size.

For each guess $Q$, the algorithm applies this routine to every component $K$ containing a separator with $w\br{V_K}>B_\rho$, and chooses the least-cost candidate $Q\cup\bigcup_K A_K$. Then, the algorithm recurses on each component. The full procedure is given by Algorithm \ref{alg:centroid-max-coverage}.

\begin{algorithm}[h]
    \DontPrintSemicolon
    \caption{Two-separator recursion with core coverage}\label{alg:centroid-max-coverage}
    \textbf{Input} A tree $T=(V,E,D,c,w)$ and $B_\rho=\rho W$ and $B_\eta=\eta W$\;
    \textbf{Output} An $\eta$-feasible supply-edge set $C\subseteq E$\;
    \textbf{Procedure} $\textsc{SeparatorMaxCoverage}(H)$\;
    \Indp
        \If{$w\br{V_H}\leq B_\eta$}{
            \textbf{return} $\emptyset$\;
        }
        Run $\textsc{Separator}(H, w, 3)$, obtaining $\mathcal{L}_H$ with $|\mathcal{L}_H|\leq2$\;
        Let $S_H\gets H[\mathcal{L}_H]$ and  $\cF_H\gets\emptyset$\;
        \ForEach{$Q\in\{\emptyset\}\cup\{\{e\}\colon e\in E_{S_H}\}$}{
            $\mathcal{K}_Q\gets\brc{K\in H\setminus Q\colon V_K\cap\mathcal{L}_H\neq\emptyset}$, $\mathcal{K}_Q^+\gets\{K\in\mathcal{K}_Q\colon w\br{V_K}>B_\rho\}$\;
            \ForEach{$K\in\mathcal{K}_Q^+$}{
                $S_K\gets K[S_H]$ and $q_K\gets w\br{V_K}-B_\rho$\;
            }
            \If{$\exists K\in\mathcal{K}_Q^+\colon w\br{V_{S_K}}>B_\rho$}{
                \textbf{continue}\;
            }
            \ForEach{$K\in\mathcal{K}_Q^+$}{
                $A_K\gets\textsc{CoreQuotaCover}(K,S_K,f_K,c_K,q_K)$\;
            }
            $F_Q\gets Q\cup\bigcup_{K\in\mathcal{K}_Q^+}A_K$ and add $F_Q$ to $\cF_H$\;
        }
        $F_H\gets \argmin_{F\in\cF_H}\brc{c(F)}$ and $C_H\gets F_H$\;
        \ForEach{$H'$ of $H\setminus F_H$}{
            $C_H\gets C_H\cup\textsc{SeparatorMaxCoverage}(H')$\;
        }
        \textbf{return} $C_H$\;
    \Indm
    $C\gets\textsc{SeparatorMaxCoverage}(T)$\;
    \textbf{return} $C$\;
\end{algorithm}

\begin{theorem}\label{thm:tree-centroid-max-coverage}
    Fix rational constants $0<\rho<\eta<1$. The Algorithm~\ref{alg:centroid-max-coverage} runs in polynomial time and returns an $\eta$-feasible supply-edge set $C\subseteq E$ with $
        c(C)\leq\cl{\ln\br{\frac{1-\rho}{\eta-\rho}}}\cdot\OPTdem{T}{\rho}.$
\end{theorem}
\begin{proof}
    Consider a recursive call for a subgraph $H$ of $T$, and let $q_H=w\br{V_H}-B_\rho>0$. 
    Define the \emph{excess} of $H$ by $\Phi(H)=\max\brc{0,w\br{V_H}-B_\rho}$.
    Let $Q$ be a feasible guess, i.e., for each $K\in\mathcal{K}_Q^+$ there exists $A\subseteq E_K\setminus E_{S_K}$ with $f_K(A)\geq q_K$. Let $F_H$ be the edges chosen by the algorithm for this guess. Consider a component $H'\in H\setminus F_H$. There are two cases.
    \begin{enumerate}
        \item $V_{S_K}\subseteq V_{H'}$ for some $K\in \mathcal{K}_Q$. Then, $H'=K\setminus U_K(A_K)$. Since $q_K=w\br{V_K}-B_\rho\leq q_H$,
    \[
            \Phi(H')=\max\brc{0,w\br{V_{H'}}-B_\rho}\leq w\br{V_{K}}-f_K\br{A_K}-B_\rho
        \leq q_K-(1-1/e)\cdot q_K
        =q_K/e\leq q_H/e.
    \]
    \item $H'$ contains no separator. Then, we have $w\br{V_{H'}}\leq w\br{V_H}/3$ and $w\br{V_H}>B_\rho$, so
    \[
        \Phi(H')=\max\brc{0,w\br{V_{H'}}-B_\rho}
        \leq\max\brc{0,w\br{V_H}/3-B_\rho}\leq(w\br{V_H}-B_\rho)/3=q_H/3<q_H/e.
    \]
    \end{enumerate} 
    Thus, if we consider the recursion tree on top of the components for which the algorithm was called, every child has excess at most $1/e$ times its parent's. Initially $\Phi(T)=W-B_\rho=(1-\rho)\cdot W$. When the recursion reaches depth
    $
        L=\left\lceil\ln\left(\frac{1-\rho}{\eta-\rho}\right)\right\rceil,
    $
    for every leaf component $H$,
    \[
        \Phi(H)\leq\frac{W-B_\rho}{e^L}\leq\frac{W-B_\rho}{e^{\ln\left(\frac{1-\rho}{\eta-\rho}\right)}}=(\eta-\rho)\cdot W=B_\eta-B_\rho,
    \]
    which makes it $\eta$-feasible.

    The cost bound is proved in Appendix~\ref{app:tree-centroid-max-coverage-cost}.
\end{proof}

Combining Theorem~\ref{thm:tree-centroid-max-coverage} with Theorem~\ref{thm:hc-gpd-transfer} gives the following corollary.
\begin{corollary}\label{cor:tree-hc-improved}
    There is a polynomial-time $\frac{4e}{e-1}$-approximation algorithm for \textsc{Hierarchical Clustering with Demands} on supply trees.
\end{corollary}
\begin{proof}
    Set $\rho=\frac12$ and $\eta=\frac{e+1}{2e}$. Then $
        \frac{1-\rho}{\eta-\rho}=e$ and $
        \cl{\ln\br{\br{1-\rho}/\br{\eta-\rho}}}
        =\left\lceil\ln e\right\rceil=1.$
    Theorem~\ref{thm:tree-centroid-max-coverage} therefore gives a polynomial-time subroutine with $\gamma_{\rho, \eta}=1$. By Theorem~\ref{thm:hc-gpd-transfer}, the resulting approximation ratio is $
        \frac{2}{1-\eta}=\frac{4e}{e-1}.$
\end{proof}

\subsection{Single-criterion approximation via Multicut}
\label{sec:Single-Criteria-Approximation-via-Multicut}

We now show how to employ a similar approach to obtain a single-criteria approximation for \textsc{Graph Partitioning with Demands} on trees via the \textsc{Partial Multicut} problem. 
In the latter problem, the input is a graph $G=\br{V, E, D, c, w}$ and an additional value $0\leq q\leq 1$. The goal is to find a subset of edges $C$, minimizing $c\br{C}$, such that the overall demand of edges $uv\in D$ having endpoints $u$ and $v$ in different connected components of $G\setminus C$ is at least $q\cdot w\br{V}$. For supply trees, Mestre~\cite{DBLP:conf/stacs/Mestre08} gives a $(2+\delta)$-approximation in polynomial-time. 

This time our algorithm is not recursive, and we pick a separator set large enough so that every connected component has demand at most $B_\rho$. Let $t=\cl{1/\rho}$ be this size.
We give a procedure $\textsc{CoreFullQuota}(K,S_K,f_K,c_K,q_K)$ that uses \textsc{Partial Multicut} on trees to find a quota-covering cut in an auxiliary tree $T_K$, and returns the corresponding set $A_K$ of original supply edges. We create this instance as follows: we forbid cutting the edges of $S_K$ (defined in the previous section) by setting their costs to $\infty$. Then, for each $K'\in K\setminus V_{S_K}$, we create a copy of $K'$. We index the two representatives of each vertex $u\in V_{K'}$ by $u^{\mathrm L}$ and $u^{\mathrm R}$ (one of them is the original vertex). Now, for each $uv\in D[K]$ if $\brc{u, v}\cap {S_K}=\emptyset$, we replace this demand edge by a cross-copy demand edge $u^{\mathrm L}v^{\mathrm R}$ with demand $w\br{u^{\mathrm L}v^{\mathrm R}}=w\br{uv}$. Other demands stay unchanged. The output consists of the original supply edges corresponding to the selected edges in $T_K$, with duplicates merged.

\begin{algorithm}[h]
    \DontPrintSemicolon
    \caption{Full-quota separator partitioning}\label{alg:centroid-full-quota}
    \textbf{Input} Tree $T=(V,E,D,c,w)$, $B_\rho=\rho W$ and $B_\eta=\eta W$, and fixed $\epsilon>0$\;
    \textbf{Output} A $\rho$-feasible supply-edge set $C\subseteq E$\;
    Run $\textsc{Separator}(T, w, t=\lceil1/\rho\rceil)$, to obtain $\mathcal{L}_T$ \;
    $S_T\gets T[\mathcal{L}_T]$ and $\cF_T\gets\emptyset$\;
    \ForEach{$Q\subseteq E_{S_T}$ with $|Q|\leq|\mathcal{L}_T|-1$}{
        $\mathcal{K}_Q\gets \brc{K\in T\setminus Q:V_K\cap\mathcal{L}_T\neq\emptyset}$ and $\mathcal{K}_Q^+\gets\{K\in\mathcal{K}_Q:w\br{V_K}>B_\rho\}$\;
        \ForEach{$K\in\mathcal{K}_Q^+$}{
            $S_K\gets K[S_T]$ and $q_K\gets w\br{V_K}-B_\rho$\;
        }
        \If{$\exists K\in\mathcal{K}_Q^+\colon w\br{V_{S_K}}>B_\rho$}{
            \textbf{continue}\;
        }
        \ForEach{$K\in\mathcal{K}_Q^+$}{
            $A_K\gets\textsc{CoreFullQuota}(K,S_K,f_K,c_K,q_K)$\;
        }
        $F_Q\gets Q\cup\bigcup_{K\in\mathcal{K}_Q^+}A_K$ and add $F_Q$ to $\cF_T$\;
    }
    \textbf{return} $C\gets\arg\min_{F\in\cF_T}\brc{c(F)}$\;
\end{algorithm}

\begin{theorem}\label{thm:tree-centroid-full-quota}
    Fix rational $0<\rho<1$ and $\epsilon>0$. The Algorithm~\ref{alg:centroid-full-quota} returns a $\rho$-feasible cut $C$ with $
        c(C)\leq(4+\epsilon)\cdot \OPTdem{T}{\rho}.$
    For fixed $\rho$ and $\epsilon$, its running time is $n^{O(1/\rho)}$.
\end{theorem}
The proof is deferred to Appendix~\ref{app:tree-centroid-full-quota}.






\section*{Generative AI usage disclosure}

Throughout this work several use cases of generative AI tools have been employed including: drafting and write-up of the manuscript, assisting with literature review, performing calculations, and generating solutions. In particular, several results were conceived while conversing with \textsc{GPT 6.1 Sol} and \textsc{GPT 6 Luna}. Our pipeline of development was as follows: firstly the authors conceived a batch of human-generated results including the approximation algorithms for \textsc{Generalized Conductance} and their use in approximating \textsc{Graph Partitioning with Demands} and \textsc{Hierarchical Clustering with Demands}. In the second stage the agentic AI was used with these contributions as a context to generate more results. The typical workflow would be to ask the agent to come up with several ideas for a proof and then, the authors would pick the most promising direction for the agent to explore. The results conceived in that manner include anchor-LP-based approximation algorithm for \textsc{Graph Partitioning with Demands} on supply trees and hardness of approximation of \textsc{Hierarchical Clustering with Demands}. These raw ideas were then checked, simplified and polished by the authors. In the last round of development the authors used a mixed approach: they would provide the agent with their own partial progress and ask it to work out the details or conversely, ask the agent to conceive ideas for a proof, but worked out the details themselves. These results include combinatorial algorithms for \textsc{Graph Partitioning with Demands} on supply trees and hardness results based on \textsc{Vertex Cover}/\textsc{Min-Sum Vertex Cover}. Several times throughout this work, AI was also used to generalize results and improve the quality of the arguments. Moreover, the authors verified and redacted all AI-generated reasoning in order to obtain a human-readable manuscript and take full responsibility for the presentation and correctness of this work.

\bibliographystyle{plain}
\bibliography{references}
\appendix

\newpage
\section{Deferred Proofs}\label{app:missing-proofs}

\subsection{Proof of Theorem~\ref{thm:gpd-extension}}\label{app:gpd-extension}
\begin{proof}[Proof of Theorem~\ref{thm:gpd-extension}]
    For a subgraph $H$ of $G$ with $w\br{V_H}\geq B_\eta$, we set $\eta_H=B_\eta/w\br{V_H}$ and $\rho_H=B_\rho/w\br{V_H}$. The algorithm uses $\eta_0=\max\brc{4/5,\eta_H}$ and we observe that $\rho_H/\eta_H=\rho/\eta$.
    Recall that $\OPTdem{H}{\rho_H}$ is the cost of an optimal solution to \textsc{Graph Partitioning with Demands} on $H$ with threshold $\rho_H$.
    It is easy to see that Algorithm~\ref{alg:recursive-partitioning} runs in polynomial time and the returned solution is $\eta$-feasible. Hence, we focus on the approximation ratio.

For a subgraph $H$ of $G$ for which the procedure has been called, let $d\br{H}$ denote the number of recursion levels the algorithm performs on $H$. Every recursive call takes $\eta_0\geq 4/5$ and, by Theorem~\ref{thm:gpd-bicriteria}, produces connected components of demand at most $\eta_0 \cdot w\br{V_H}\leq \frac{4}{5} w\br{V_H}$ each.
Hence, for every `child component' $H'$ of $H\setminus C$ in the recursion tree, we get $\eta_{H'}\geq 5/4\cdot \eta_H$.
Because the last partitioning step happens when $\eta_0=\eta_H\geq 4/5$, the recursion depth is bounded by
\[
    d\br{H}\leq \cl{\log_{5/4}\br{w\br{V_H}/B_\eta}}
\]
due to the fact that after reaching $\eta_{H}\geq 4/5$ the algorithm performs one last partitioning step, i.e., the last recursive call. Since $\cl{\log_{5/4}\br{w\br{V_H}/B_\eta}}$ is the number of iterations required to reach $\eta_H\geq 1$, which is one more than we need to reach $\eta_H\geq 4/5$, we get that the above quantity is correct.

Now, we prove by induction on $d\br{H}$ that Algorithm~\ref{alg:recursive-partitioning} returns a subset of edges of cost at most $d\br{H}\cdot \kappa\br{\rho,\eta}\cdot \gamma\br{n}\cdot \OPTdem{H}{\rho_H}$.

    \emph{Base cases, $d\br{H}=0$ and $d\br{H}=1$.} If $d\br{H}=0$, then the cost is $0$ since the algorithm terminates without any partitioning. If $d\br{H}=1$, then $B_\eta/w\br{V_H}\geq 4/5$, so $\eta_0=\eta_H=B_\eta/w\br{V_H}$, and after a single partitioning step for each connected component $H'$ of $H\setminus C$, $w\br{V_{H'}}\leq \eta_0 \cdot w\br{V_H}=B_\eta$. Hence the algorithm terminates at the beginning of the next recursive step without adding any supply edges to the cut. By the explicit bound established in the proof of Theorem~\ref{thm:gpd-bicriteria}, the cost of the solution is bounded by 
    \[
    \frac{\gamma\br{n}}{\br{\eta_H+\delta_{\rho_H}-1}^2}\cdot \OPTdem{H}{\rho_H}.
    \]
    We establish the bound 
    \[
    \frac{1}{\br{\eta_H+\delta_{\rho_H}-1}^2}\leq \kappa\br{\rho,\eta}
    \]
    in two cases:
    \begin{enumerate}
        \item If $\rho_H\leq 3/4$, then $\delta_{\rho_H}=1/4$ and $\eta_H+\delta_{\rho_H}-1\geq 4/5+1/4-1=1/20$, using $\eta_H\geq 4/5$. Hence 
        \[\frac{1}{\br{\eta_H+\delta_{\rho_H}-1}^2}\leq 400\leq \kappa\br{\rho,\eta}.\]
        \item If $\rho_H> 3/4$, then $\delta_{\rho_H}=1-\rho_H$ and, using $\rho_H=\br{\rho/\eta}\cdot\eta_H$ and $\eta_H\geq 4/5$,
        \[
            \eta_H+\delta_{\rho_H}-1=\eta_H-\rho_H=\br{1-\rho/\eta}\cdot\eta_H\geq \br{1-\rho/\eta}\cdot\frac{4}{5}.
        \]
        Hence 
        \[
        \frac{1}{\br{\eta_H+\delta_{\rho_H}-1}^2}\leq \frac{1}{\br{\br{1-\rho/\eta}\cdot 4/5}^2}\leq \kappa\br{\rho,\eta}.
        \]
    \end{enumerate}
    In both cases the cost is at most 
    \[\kappa\br{\rho,\eta}\cdot\gamma\br{n}\cdot \OPTdem{H}{\rho_H}=d\br{H}\cdot \kappa\br{\rho,\eta}\cdot\gamma\br{n}\cdot\OPTdem{H}{\rho_H}.\]
 
    \emph{Inductive step.} Suppose the claim holds for every graph $H'$ with $d\br{H'}\leq d-1$, and let $H$ satisfy $d\br{H}=d\geq 2$. Let $C$ be the cut computed when $H$ is the input. It holds $B_\eta/w\br{V_H}<4/5$, so $\eta_0=4/5$, and for every $H'$ in $H\setminus C$, 
    \[
    w\br{H'}\leq \eta_0\cdot w\br{V_H}=\frac{4}{5}\cdot w\br{V_H},\]
    thus $d\br{H'}\leq d-1$. By Theorem \ref{thm:gpd-bicriteria}, the cost introduced by the current partitioning step is at most 
    \[\frac{\gamma\br{n}}{\br{\eta_0+\delta_{\rho_H}-1}^2}\cdot\OPTdem{H}{\rho_H},
    \]
     and we establish the bound 
    \[\frac{1}{\br{\eta_0+\delta_{\rho_H}-1}^2}\leq \kappa\br{\rho,\eta}\] similarly as before:
    \begin{enumerate}
        \item If $\rho_H\leq 3/4$, then $\delta_{\rho_H}=1/4$ and $\eta_0+\delta_{\rho_H}-1=4/5+1/4-1=1/20$, so the term is at most $400\leq \kappa\br{\rho,\eta}$.
        \item If $\rho_H> 3/4$, then $\delta_{\rho_H}=1-\rho_H$ and, using $\rho_H=\br{\rho/\eta}\cdot\eta_H<\br{\rho/\eta}\cdot\eta_0$,
        \[
            \eta_0+\delta_{\rho_H}-1=\eta_0-\rho_H> \eta_0-\br{\rho/\eta}\cdot\eta_0=\br{1-\rho/\eta}\cdot\frac{4}{5},
        \]
        so the term is at most 
        \[
        \frac{1}{\br{\br{1-\rho/\eta}\cdot 4/5}^2}\leq \kappa\br{\rho,\eta}.
        \]
    \end{enumerate}
    Thus the top level contributes at most $\kappa\br{\rho,\eta}\cdot\gamma\br{n}\cdot\OPTdem{H}{\rho_H}$. By the inductive hypothesis, for each connected component $H'$ of $H\setminus C$, 
    \[c\br{C_{H'}}\leq \br{d-1}\cdot \kappa\br{\rho,\eta}\cdot\gamma\br{n}\cdot\OPTdem{H'}{\rho_{H'}}.\] 
    Since the components in $H\setminus C$ are vertex and edge-disjoint and the $C^*\cap E_{H'}$ is a valid solution for \textsc{Graph Partitioning with Demands} on $H'$ in $H\setminus C$ with $\rho_{H'}$, we have that 
    \[\OPTdem{H'}{\rho_{H'}}\leq c\br{C^*\cap E_{H'}}.
    \]
    Moreover, by summing up over all $H'$ we get
    \[
    \sum_{H'\in H\setminus C}\OPTdem{H'}{\rho_{H'}}\leq \OPTdem{H}{\rho_H},
    \]
    and therefore
    \begin{align*}
        c\br{C}&=c\br{C_H}+\sum_{H'\in H\setminus C}c\br{C_{H'}}
        \\&\leq \kappa\br{\rho,\eta}\cdot\gamma\br{n}\cdot\OPTdem{H}{\rho_H}+\br{d-1}\cdot\kappa\br{\rho,\eta}\cdot \gamma\br{n}\sum_{H'\in H\setminus C}\OPTdem{H'}{\rho_{H'}}
        \\&\leq 
        d\cdot \kappa\br{\rho,\eta}\cdot\gamma\br{n}\cdot\OPTdem{H}{\rho_H}.
    \end{align*}
    Therefore, the induction is complete. Applying it to $H=G$ yields the statement of the theorem.
\end{proof}

\subsection{Proof of Lemma~\ref{lem:tree-separator}}\label{app:tree-separator}
\begin{proof}[Proof of Lemma~\ref{lem:tree-separator}]
    Root $H$ arbitrarily and initialize the current root-containing subtree $R$ to $H$. While $w\br{V_R}>w\br{V_H}/t$, choose the deepest vertex $x$ such that the subtree $R_x$ rooted at $x$ in the current $R$ satisfies $w\br{V_{R_x}}>w\br{V_H}/t$.
    Add $x$ to $\mathcal L_H$ and remove $R_x$ from $R$. The peeled subtrees are pairwise disjoint, so the sum of their internal demand weights is at most $w\br{V_H}$. Since every peeled subtree has demand greater than $w\br{V_H}/t$, there are at most $t-1$ of them. By the choice of a deepest such vertex, every child subtree $T'$ of $x$ in $R_x$ satisfies $w\br{V_{T'}}\leq w\br{V_H}/t$. Therefore, after removing all separators, every connected component $K$ in $H\setminus\mathcal L_H$ is one of these child subtrees at the iteration when its parent separator was selected, or the final remainder $R$. The final remainder satisfies $w\br{V_R}\leq w\br{V_H}/t$ by the stopping condition. In either case, $w\br{V_K}\leq w\br{V_H}/t$ as desired.
\end{proof}

\subsection{Cost bound in Theorem~\ref{thm:tree-centroid-max-coverage}}\label{app:tree-centroid-max-coverage-cost}
\begin{proof}[Proof of the cost bound in Theorem~\ref{thm:tree-centroid-max-coverage}]
For every component $H$ in the recursion, define its local threshold by
\[
    \rho_H=
    \begin{cases}
        \min\brc{1,\frac{B_\rho}{w\br{V_H}}}, & w\br{V_H}>0,\\
        1, & w\br{V_H}=0.
    \end{cases}
\]
Every nonterminal recursive call has $w\br{V_H}>B_\eta>B_\rho$, so $\rho_H=B_\rho/w\br{V_H}$. Thus, $\rho_H\cdot w\br{V_H}=\min\brc{B_\rho,w\br{V_H}}$. This is equivalent to the global budget $B_\rho$. In particular, components with $w\br{V_H}\leq B_\rho$ need no further cut. We prove by induction on the maximum recursion depth $d(H)$ for $H$ that the total cost $C_H$ returned below $H$ satisfies
\[
    c(C_H)\leq d(H)\cdot\OPT_{\rho_H}(H).
\]
For a terminal component, $d(H)=0$ and $C_H=\emptyset$, so the claim holds. Otherwise, let $F_H$ be the cut selected at $H$, and let $H_1,\ldots,H_m$ be the components of $H\setminus F_H$. Then
\[
    d(H)=1+\max_{1\leq i\leq m}d(H_i).
\]

Let $F_H^*$ be an optimal $\rho_H$-feasible cut of $H$, and set $Q^*=F_H^*\cap E_{S_H}$. The algorithm considers this guess. For each $K\in\mathcal K_{Q^*}^+$, the component of $K\setminus F_H^*$ containing $S_K$ has internal demand at most $B_\rho$. Hence the restriction of $F_H^*$ to $K$ separates at least $q_K=w(V_K)-B_\rho$ demand from $S_K$. It follows that
\[
    \OPT_{q_K}(K)\leq c(F_H^*\cap E_K).
\]
Since $\textsc{CoreQuotaCover}$ returns a set $A_K$ with $f_K(A_K)\geq(1-1/e)q_K$ and $c_K(A_K)\leq\OPT_{q_K}(K)$, the candidate for $Q^*$ has cost at most
\[
    c(Q^*)+\sum_{K\in\mathcal K_{Q^*}^+}\OPT_{q_K}(K)
    \leq c(F_H^*)=\OPT_{\rho_H}(H).
\]
Since the algorithm chooses the least-cost candidate, $c(F_H)\leq\OPT_{\rho_H}(H)$.

For each child $H_i$, restricting $F_H^*$ to $E_{H_i}$ is feasible for $\OPT_{\rho_{H_i}}(H_i)$. The child edge sets are pairwise disjoint, so
\[
    \sum_{i=1}^m\OPT_{\rho_{H_i}}(H_i)
    \leq\sum_{i=1}^m c(F_H^*\cap E_{H_i})
    \leq\OPT_{\rho_H}(H).
\]
Assuming the induction claim for each child, we obtain
\begin{align*}
    c(C_H)
    &=c(F_H)+\sum_{i=1}^m c(C_{H_i})\\
    &\leq\OPT_{\rho_H}(H)+(d(H)-1)\cdot \sum_{i=1}^m\OPT_{\rho_{H_i}}(H_i)\\
    &\leq d(H)\cdot\OPT_{\rho_H}(H).
\end{align*}
By the excess-demand bound in the proof of Theorem~\ref{thm:tree-centroid-max-coverage}, all recursive calls terminate by depth
\[
    L=\left\lceil\ln\left(\frac{1-\rho}{\eta-\rho}\right)\right\rceil.
\]
Thus $d(T)\leq L$. If $W>0$, then $\rho_T=B_\rho/W=\rho$, so $\OPT_{\rho_T}(T)=\OPTdem{T}{\rho}$ and the claimed cost bound follows. If $W=0$, the algorithm returns the empty cut and the bound is immediate.
\end{proof}

\subsection{Proof of Theorem~\ref{thm:tree-centroid-full-quota}}\label{app:tree-centroid-full-quota}
\begin{proof}[Proof of Theorem~\ref{thm:tree-centroid-full-quota}]
    By Lemma~\ref{lem:tree-separator}, $\spr{\mathcal{L}_T}\leq t-1$, and for every $H\in T\setminus\mathcal L_T$, $w\br{H}\leq W/t\leq B_\rho$. By construction, all components $H\in T\setminus C$ have total demand bounded by $w\br{V_H}\leq B_\rho$.

    Let $F^*$ be an optimal $\rho$-feasible cut and set $Q^*=F^*\cap E_{S_T}$. For each $K\in\mathcal K_{Q^*}^+$, fix an optimal $A_K^*$ for $\OPT_{q_K}(K)$. To obtain a valid multicut $M$ for $T_K$, for each edge $e=uv\in A_K^*$ include both representative edges $u^{\mathrm L}v^{\mathrm L}$ and $u^{\mathrm R}v^{\mathrm R}$ in $M$. This doubles the cost of each edge, so $c(M)\leq2\cdot \OPT_{q_K}(K)$. Mestre's $(2+\delta)$-approximation returns a multicut $\widehat A_K$ with
    \[
        c(\widehat A_K)\leq(2+\delta)\cdot c(M)\leq(4+\epsilon)\cdot\OPT_{q_K}(K),
    \]
    where $\delta=\epsilon/2$. The routine returns $A_K$ by replacing each selected copy of an edge in $\widehat A_K$ by its corresponding original supply edge and merging duplicates. This preserves the separated demand and cannot increase the cost, so $c(A_K)\leq c(\widehat A_K)\leq(4+\epsilon)\cdot\OPT_{q_K}(K)$.

    The restrictions of $F^*$ to the disjoint components in $\mathcal K_{Q^*}^+$ are feasible for their quota problems, and their edges together with $Q^*$ are a subset of $F^*$. Therefore,
    \[
        c(Q^*)+\sum_{K\in\mathcal K_{Q^*}^+}\OPT_{q_K}(K)\leq\OPTdem{T}{\rho}.
    \]
    Since the algorithm considers $Q^*$ and chooses a minimum-cost candidate, its output satisfies
    \[
        c(C)\leq c(Q^*)+\sum_{K\in\mathcal K_{Q^*}^+}c(A_K)
        \leq(4+\epsilon)\cdot\OPTdem{T}{\rho}.
    \]

    Since the number of sets $Q$ we enumerate is bounded by $\sum_{i=0}^{t-1} \binom{|E_{S_T}|}{i}=n^{O(1/\rho)}$, the algorithm runs in polynomial time for any fixed choice of $\rho$.
\end{proof}

\section{Hardness Results}\label{sec:gpd-hardness}

In this section we provide several approximation-hardness results for \textsc{Generalized Conductance}, \textsc{Graph Partitioning with Demands} and \textsc{Hierarchical Clustering with Demands} based on the Small-Set Expansion Hypothesis, UGC, and (for \textsc{Generalized Conductance}) the cited uniform-demand \textsc{Sparsest Cut} hardness claim. The tree bounds are stated under UGC, and the \textsc{Generalized Conductance} bound is conditional on the sparsest-cut claim. The cited OpenAI Unique Games Theorem is not treated as an established proof of UGC: its Lean verification contains an admitted \texttt{sorry}, and the cited sparsest-cut hardness result has not been verified in Lean. Thus the corresponding NP-hardness and ``unless $P=NP$'' conclusions follow only if the relevant assumption is true. The \emph{Small-Set Expansion Hypothesis} states that, for every constant $\zeta>0$, there is a constant $\delta>0$ for which it is NP-hard, given a regular graph, to distinguish between the existence of a vertex set of normalized size $\delta$ and expansion at most $\zeta$, and the condition that every vertex set of normalized size $\delta$ has expansion at least $1-\zeta$. Some reductions in this section use rational numbers instead of natural ones, however this can be fixed by an appropriate scaling.

\subsection{Hierarchical Clustering with Demands on general graphs}
For every unweighted undirected graph $G=(V,E)$ used as a source graph below, write
\[
    \operatorname{vol}_G(S)=\sum_{u\in S}\deg_G(u),
    \qquad
    M_G=\operatorname{vol}_G(V)=2|E|,
\]
\[
    \mu_G(u)=\frac{\deg_G(u)}{M_G},
    \qquad
    \mu_G(S)=\sum_{u\in S}\mu_G(u)=\frac{\operatorname{vol}_G(S)}{M_G},\qquad
    \Phi_G(S)=\frac{|\delta_G(S)|}{\operatorname{vol}_G(S)},
\]
where $\deg_G(u)$ is the degree of $u$ in $G$.
The transition probabilities in a random walk on $G$ are given by $P_G(u,v)=1/\deg_G(u)$ if $uv\in E$ and $P_G(u,v)=0$ otherwise. Its stationary distribution is $\mu_G(u)=\deg_G(u)/M_G$: for every $v\in V$,
\[
    \sum_{u\in V}\mu_G(u)\cdot P_G(u,v)
    =\sum_{u\colon uv\in E}\frac{1}{M_G}
    =\frac{\deg_G(v)}{M_G}
    =\mu_G(v).
\]
Now, set
\[
    p_G(u,v)=\mu_G(u)\cdot P_G(u,v).
\]
Thus $p_G(u,v)=p_G(v,u)=1/M_G$ on every edge. Then, for any $S\subseteq V$ we have
\[
    \sum_{uv\in\delta_G(S)}p_G(u,v)
    =\frac{|\delta_G(S)|}{M_G}
    =\mu_G(S)\cdot \Phi_G(S),
    \qquad
    \sum_{uv\in E}p_G(u,v)=\frac12.
\]

We also recall the Gaussian expansion profile. For $r\in(-1,1)$ and $x\in(0,1)$, let
\[
    \begin{pmatrix}Z\\Z'\end{pmatrix}
    \sim\mathcal{N}\left(
        \begin{pmatrix}0\\0\end{pmatrix},
        \begin{pmatrix}1&r\\r&1\end{pmatrix}
    \right).
\]
Thus $Z$ and $Z'$ are standard Gaussian random variables with correlation $r$. Let $t_x\in\mathbb{R}$ satisfy $\Pr[Z\geq t_x]=x$. We define the expansion profile of the Gaussian noise graph $\mathcal{G}(r)$ by
\[
    \Phi_{\mathcal{G}(r)}(x)
    =\Pr\left[Z'<t_x\mid Z\geq t_x\right]
    =1-\frac{\Pr[Z\geq t_x,Z'\geq t_x]}{x}.
\]

For any graph $Y=(V_Y,E_Y)$ define $
    \alpha_Y=\max_{u\in V_Y}\brc{\mu_Y(u)}.
$
We use the following result of Raghavendra, Steurer, and Tulsiani~\cite{ReductionsBetweenExpansionProblems}, noting that the standard product-space construction in the source theorem with tensor power $R\geq2$ certifies that $\alpha_Y\leq q^{-R}\leq1/4$.

\begin{theorem}[Raghavendra-Steurer-Tulsiani]\label{thm:rst-global-expansion}
    For every integer $q\geq2$, there are constants $K>0$ and $\varepsilon_0\in(0,1]$ such that for every $\varepsilon\in(0,\varepsilon_0)$ and every $\gamma>0$, assuming the Small-Set Expansion Hypothesis, it is hard to distinguish between the following cases for a graph $H=(V_H,E_H)$:
    \begin{description}
        \item[YES:] There is a partition $S_1,\ldots,S_q$ of $V_H$ such that, for every $i\in[q]$, 
            $\mu_H(S_i)=\frac{1}{q}$ and $\Phi_H(S_i)\leq K\cdot\varepsilon$.
        \item[NO:] Every $S\subseteq V_H$ satisfies $
            \Phi_H(S)
            >
            \Phi_{\mathcal{G}(1-\varepsilon/2)}(\mu_H(S))
            -\frac{\gamma}{\mu_H(S)}$.
    \end{description}
    The output graph can be chosen so that $\alpha_H\leq q^{-R}\leq1/4$ for its tensor power $R\geq2$.
\end{theorem}

We derive from Theorem~\ref{thm:rst-global-expansion} the form of the gap used in our reduction. This step is relatively standard, however we include its proof for completeness.

\begin{lemma}\label{lem:globalized-long-code-gap}
    Fix an integer $q\geq2$ and a constant $b\in(0,1/2]$. Assuming the Small-Set Expansion Hypothesis, there exist constants $K>0$ and $a_b>0$ such that for a sufficiently small choice of $\varepsilon>0$ it is hard to distinguish between the following cases for a graph $Y=(V_Y,E_Y)$:
    \begin{description}
        \item[YES:] There is a partition $S_1,\ldots,S_q$ of $V_Y$
        such that, for every $i\in[q]$, $
            \mu_Y(S_i)=\frac{1}{q}$ and $
            \Phi_Y(S_i)\leq K\cdot\varepsilon$.
        \item[NO:] For every $S\subseteq V_Y$, where $b\leq\mu_Y(S)\leq\frac{1}{2}$, $\Phi_Y(S)\geq a_b\cdot \sqrt{\varepsilon}$.
    \end{description}
    The output graph also satisfies $\alpha_Y\leq q^{-R}\leq1/4$ for its tensor power $R\geq2$.
\end{lemma}
\begin{proof}
    In the YES case of Theorem~\ref{thm:rst-global-expansion}, the statement stays the same.

    Assume the NO case.
    Write $\varphi$ and $\Phi_{\mathrm N}$ for the standard Gaussian density and CDF, respectively, and let $t_b$ satisfy $\Pr[Z\geq t_b]=b$. For $x\in[b,1/2]$, we have $t_x\in[0,t_b]$. Set $\Delta=Z'-Z$ and $r=1-\varepsilon/2$. Then
    \[
        \operatorname{Var}(\Delta)=2\cdot(1-r)=\varepsilon,
        \qquad
        \operatorname{Cov}(Z,\Delta)=r-1=-\frac{\varepsilon}{2}.
    \]
    Since $\Delta$ is a linear combination of $Z$ and $Z'$, the pair $(Z,\Delta)$ is jointly Gaussian. For centered jointly Gaussian variables, the conditional distribution is Gaussian with
    \[
        Z\mid\Delta=d\sim\mathcal{N}\left(
            \frac{\operatorname{Cov}(Z,\Delta)}{\operatorname{Var}(\Delta)}\cdot d,\,
            1-\frac{\operatorname{Cov}(Z,\Delta)^2}{\operatorname{Var}(\Delta)}
        \right).
    \]
    Taking $d=-u$ gives $Z\mid\Delta=-u\sim\mathcal{N}(u/2,1-\varepsilon/4)$, where the second parameter is the variance. On this event, $Z'=Z-u$, so the crossing event $\{Z\geq t_x,\ Z'<t_x\}$ is $\{t_x\leq Z<t_x+u\}$. Therefore,
    \[
        \Phi_{\mathcal{G}(1-\varepsilon/2)}(x)
        =\frac{1}{x}\cdot\int_0^\infty\frac{\varphi(u/\sqrt{\varepsilon})}{\sqrt{\varepsilon}}\cdot\left[
            \Phi_{\mathrm N}\left(\frac{t_x+u/2}{\sqrt{1-\varepsilon/4}}\right)
            -\Phi_{\mathrm N}\left(\frac{t_x-u/2}{\sqrt{1-\varepsilon/4}}\right)
        \right]\,du.
    \]
    Substituting $u=\sqrt{\varepsilon}s$ in this integral gives
    \[
        \Phi_{\mathcal{G}(1-\varepsilon/2)}(x)
        =\frac{1}{x}\cdot\int_0^\infty\varphi(s)\cdot\left[
            \Phi_{\mathrm N}\left(\frac{t_x+\sqrt{\varepsilon}\cdot s/2}{\sqrt{1-\varepsilon/4}}\right)
            -\Phi_{\mathrm N}\left(\frac{t_x-\sqrt{\varepsilon}\cdot s/2}{\sqrt{1-\varepsilon/4}}\right)
        \right]\,ds.
    \]
    For each fixed $s$, the bracket divided by $\sqrt{\varepsilon}$ converges, uniformly for $t_x\in[0,t_b]$, to $s\cdot\varphi(t_x)$. Moreover, for $\varepsilon\leq1$ its absolute value is at most $C\cdot\sqrt{\varepsilon}\cdot s$ for an absolute constant $C$, by the boundedness of the Gaussian density and $\sqrt{1-\varepsilon/4}\geq\sqrt{3}/2$. Since $s\cdot\varphi(s)$ is integrable and $\int_0^\infty s\cdot\varphi(s)\,ds=1/\sqrt{2\pi}$, dominated convergence gives the uniform limit
    \[
        \lim_{\varepsilon\to0}\sup_{x\in[b,1/2]}
        \left|
            \frac{\Phi_{\mathcal{G}(1-\varepsilon/2)}(x)}{\sqrt{\varepsilon}}
            -\frac{\varphi(t_x)}{x\cdot\sqrt{2\pi}}
        \right|=0.
    \]
    The coefficient is continuous and positive on the compact interval $[b,1/2]$. Define
    \[
        c_b=\min_{x\in[b,1/2]}\frac{\varphi(t_x)}{x\cdot\sqrt{2\pi}}.
    \]
    Then $c_b>0$, and for all sufficiently small $\varepsilon$,
    $\Phi_{\mathcal{G}(1-\varepsilon/2)}(x)\geq a'_b\cdot\sqrt{\varepsilon}$ uniformly on this interval, where $a'_b=c_b/2$.

    Invoke Theorem~\ref{thm:rst-global-expansion} with
    $\gamma=b\cdot a'_b\cdot\sqrt{\varepsilon}/2$. For every $S$ with $x=\mu_Y(S)\in[b,1/2]$, we have $\gamma/x\leq\gamma/b$, so the NO case gives
    \[
        \Phi_Y(S)>
        \Phi_{\mathcal{G}(1-\varepsilon/2)}(x)-\frac{\gamma}{x}
        \geq a'_b\cdot\sqrt{\varepsilon}-\frac{\gamma}{b}
        =\frac{a'_b}{2}\cdot\sqrt{\varepsilon}.
    \]
    The claim follows with $a_b=a'_b/2$.
\end{proof}

\subsection{Hardness of Hierarchical Clustering with Demands}


\begin{theorem}\label{thm:sse-hcd-hardness}
    Assuming the Small-Set Expansion Hypothesis, for every constant $A\geq 1$ it is hard to approximate \textsc{Hierarchical Clustering with Demands} within a factor of $A$.
\end{theorem}
\begin{proof}
    Fix $A\geq 1$ and choose $\varepsilon>0$ satisfying assumption of Lemma~\ref{lem:globalized-long-code-gap}, and
    \[
        \varepsilon<\left(\frac{a_{1/4}}{16\cdot A\cdot (K+1)}\right)^2.
    \]
    Let $q$ be the least integer satisfying $q^{-2}\leq\varepsilon$. Apply Lemma~\ref{lem:globalized-long-code-gap} with this $q$, $b=1/4$, and $\varepsilon$, and denote its output graph by $Y=(V_Y,E_Y)$.

    Construct an HCD instance whose supply graph is $Y$ with capacity $c(uv)=p_Y(u,v)$ on each edge $uv\in E_Y$. This is well-defined because $p_Y(u,v)=p_Y(v,u)$. Let the demand graph be complete on distinct vertex pairs, with $w(uv)=\mu_Y(u)\cdot \mu_Y(v)$. By definition, for every $S\subseteq V_Y$, we have $
        c(\delta_Y(S))=\mu_Y(S)\cdot \Phi_Y(S).$
    For every cluster $X\subseteq V_Y$,
    \begin{equation}\label{eq:hcd-stationary-demand}
        w(X)=\sum_{uv\in D[X]}\mu_Y(u)\cdot \mu_Y(v)
        =\frac12\cdot \left(\mu_Y(X)^2-\sum_{u\in X}\mu_Y(u)^2\right).
    \end{equation}
    The shared normalization gives total supply capacity $c(E_Y)=1/2$, and we can assume that $\alpha_Y\leq1/4$ by Theorem~\ref{thm:rst-global-expansion}.

    \emph{YES case.} Make $S_1,\ldots,S_q$ the children of the root of $\cT$, and refine each $S_i$ to singleton leaves. This tree may not be binary, however as noted before, it can be converted to a binary tree without increasing the cost.
    Summing over the supply edges crossing the cut corresponding to the root partition, we have that
    \[
        c(C_{\mathrm{root}})=\frac12\cdot\sum_{i=1}^q c(\delta_Y(S_i))
        \leq\frac{K\cdot\varepsilon}{2}.
    \]
    By~\eqref{eq:hcd-stationary-demand}, $w(V_Y)\leq\mu_Y\br{V_Y}^2/2=1/2$, and hence the root partition contributes at most $K\cdot\varepsilon/4$ to $c\br{\cT}$. Now, observe that every supply edge internal to some $S_i$ is charged at a cluster of demand at most $w(S_i)\leq\mu_Y(S_i)^2/2\leq1/(2q^2)$. Since the total supply capacity is $1/2$, refining the $S_i$'s contributes at most $1/(4q^2)\leq\varepsilon/4$ to $c\br{\cT}$. Therefore,
    \[
        \operatorname{OPT}_{\mathrm{YES}}\leq c\br{\cT}
        \leq\frac{K+1}{4}\cdot \varepsilon.
    \]

    \emph{NO case.} Consider any binary clustering tree $\cT$. Starting at its root, follow a larger child of $\cT$ (measured by $\mu_Y(\cdot)$) until first reaching a cluster $S$ with $\mu_Y(S)\leq1/2$. Its parent $P$ has $\mu_Y(P)>1/2$, while the choice of a larger child ensures $1/4<\mu_Y(S)\leq1/2$. The NO case of Lemma~\ref{lem:globalized-long-code-gap}, with $b=1/4$, gives
    \[
        c(\delta_Y(S))=\mu_Y(S)\cdot \Phi_Y(S)
        \geq\frac{a_{1/4}}{4}\cdot \sqrt{\varepsilon}.
    \]
    Using~\eqref{eq:hcd-stationary-demand}, $\sum_{u\in P}\mu_Y(u)^2\leq\alpha_Y\cdot \mu_Y(P)$ and $\alpha_Y\leq 1/4$ we get
    \[
        w(P)\geq\frac{\mu_Y(P)\cdot(\mu_Y(P)-\alpha_Y)}2
        >\frac1{16}.
    \]
    Every supply edge crossing $S$ has its least common ancestor at $P$ or an ancestor of $P$, and thus is charged at a cluster of demand at least $w(P)$. Therefore, for every tree $\mathcal{T}$, the edges crossing $S$ alone contribute
    \[
        c(\mathcal{T})
        \geq\sum_{uv\in\delta_Y(S)}c(uv)\cdot
          w\br{V_{\operatorname{lca}_{\mathcal{T}}(u,v)}}
        \geq c(\delta_Y(S))\cdot w(P)
        \geq\frac{1}{16}\cdot\frac{a_{1/4}}{4}\cdot\sqrt{\varepsilon}
        =\frac{a_{1/4}}{64}\cdot\sqrt{\varepsilon}.
    \]
    Taking the minimum over all trees gives $\operatorname{OPT}_{\mathrm{NO}}\geq a_{1/4}\cdot\sqrt{\varepsilon}/64$.

    By our choice of $\varepsilon$, we have
    \[
        A\cdot\frac{K+1}{4}\cdot\varepsilon
        <\frac{a_{1/4}}{64}\cdot\sqrt{\varepsilon},
    \]
    so the YES and NO cases are separated by a factor greater than $A$. Therefore, any $A$-approximation algorithm for \textsc{Hierarchical Clustering with Demands} would distinguish the YES and NO cases, which is SSEH-hard.
\end{proof}

\begin{corollary}\label{cor:sse-gpd-bicriteria-hardness}
    Assuming the Small-Set Expansion Hypothesis, for every fixed $0<\rho<\eta<1$ and every constant $A\geq 1$, there is no polynomial-time $(A,\eta)$-bi-criteria approximation for \textsc{Graph Partitioning with Demands} with target $\rho$, even when the demand graph is complete.
\end{corollary}
\begin{proof}
    Suppose such an algorithm exists. Applying Theorem~\ref{thm:hc-gpd-transfer} with these $\rho$ and $\eta$ would give a polynomial-time
    $\frac{\lceil1/\rho\rceil \cdot A}{1-\eta}$-approximation for \textsc{Hierarchical Clustering with Demands}. This factor is constant for fixed $\rho$, $\eta$, and $A$, contradicting Theorem~\ref{thm:sse-hcd-hardness}.
\end{proof}

\subsection{Hardness of Generalized Conductance}
The hardness of \textsc{Generalized Conductance} essentially follows from the previous hardness results for GPD. A recent OpenAI preprint claims a constant-factor inapproximability for uniform-demand \textsc{Sparsest Cut}~\cite{OAI:Constant-factor-hardness-of-uniform-sparsest-cut-September-24-2026}, if correct, this translates to the same result for \textsc{Generalized Conductance} with uniform demands.
\begin{corollary}\label{cor:generalized-conductance-hardness}
  If the cited uniform-demand \textsc{Sparsest Cut} hardness claim is correct, then for every fixed constant $A>1$, it is NP-hard to approximate \textsc{Generalized Conductance} within a factor of $A$, even when $D=\binom{V}{2}$ and $w(e)=1$ for every $e\in D$.
\end{corollary}
\begin{proof}
  Let $G=(V,E,c)$ be a supply graph, and recall that the uniform-demand \textsc{Sparsest Cut} objective is
  \[
\phi\br{S}
=\frac{c\br{S,\bar{S}}}{\spr{S}\cdot\spr{\bar{S}}}.
  \]
  Consider the same supply graph as an instance of \textsc{Generalized Conductance} with $D=\binom{V}{2}$ and $w(e)=1$ for every $e\in D$. For a nontrivial cut $S$, writing $s=\spr{S}$, gives
  \[
w\br{S,V}=\binom{n}{2}-\binom{n-s}{2}
  =\frac{s\cdot(2n-s-1)}{2}.
  \]
  Since $1\leq s\leq n-1$, we have $n\cdot s/2\leq w\br{S,V}\leq n\cdot s$,  and the same bounds hold for $\bar{S}$. Therefore,
  \[
\frac{\phi\br{S}}{n^2}
\leq \psi_w\br{S}
\leq \frac{4\cdot \phi\br{S}}{n^2}.
  \]
  Thus, a $\gamma$-approximation for \textsc{Generalized Conductance} on these instances gives a $4\gamma$-approximation for uniform-demand \textsc{Sparsest Cut}. If the cited sparsest-cut hardness claim is correct, this proves the conditional hardness statement.
\end{proof}

\subsection{Hardness of Graph Partitioning with Demands on Stars}

The following statement gives an approximation-preserving reduction from \textsc{Vertex Cover} for the case when the supply graph is a star.

\begin{theorem}\label{thm:tree-gpd-vertex-cover-hardness}
Assuming UGC, for every fixed rational $\rho\in(0,1)$ and $\varepsilon>0$, \textsc{Graph Partitioning with Demands} is hard to approximate within a factor of $2-\varepsilon$, even when the supply graph is a star.
\end{theorem}
\begin{proof}
Let $G=(U,F)$ be a \textsc{Vertex Cover} instance. Let $V=U\cup\{r,d\}$. Construct a supply star $E$ with center $r$, leaves $U$, and one additional leaf $d$. Set $c(rd)=\infty$ and $c(ru)=1$ for every $u\in U$. The edges of the demand graph are $F\cup\{rd\}$ where $w\br{rd}=\frac{\rho \cdot \spr{F}}{1-\rho}$ and, for every $uv\in F$, $w\br{uv}=1$.

The total demand is $W=w\br{rd}+\spr{F}$, and hence $B_\rho=\rho\cdot W=w\br{rd}$. Any solution of cost at most $\spr{U}$ leaves $rd$ uncut, otherwise its cost becomes $\infty$. For any solution $C$, let $H\in T\setminus C$ be the component containing $r$. If $H$ contained both endpoints of some edge $uv\in F$, its internal demand would exceed $w\br{rd}$. Therefore, the set $C_H=V\setminus V_H$ forms a vertex cover. Moreover, we have that $c\br{C}=\spr{C_H}$ which means that the cost of any solution for the \textsc{Graph Partitioning with Demands} instance equals the size of some vertex cover. Conversely, any cover $C_H\subseteq U$ gives rise to a solution $C$ by cutting all edges incident to $C_H$, which has cost $\spr{C_H}$. Thus the reduction preserves approximation ratios.

Khot and Regev show that, assuming UGC, \textsc{Vertex Cover} is hard to approximate within a factor of $2-\varepsilon$~\cite{VertexCoverMightBeHardToApproximateToWithin2MinusEpsilon}. Since the reduction preserves approximation ratios, the claimed UGC-conditional hardness follows.
\end{proof}

\subsection{Hardness of Hierarchical Clustering with Demands on Stars}
A similar result for GPD follows from the hardness of \textsc{Min-Sum Vertex Cover} and gives the following.
\begin{theorem}\label{thm:hc-star-hardness}
Assuming UGC, \textsc{Hierarchical Clustering with Demands} is hard to approximate within a factor of $1.0748$,
even when the supply graph is a star and all capacities and demands are uniform.
\end{theorem}
\begin{proof}
The reduction is from \textsc{Min-Sum Vertex Cover} on general graphs. In this problem we are given an undirected graph $G=(U,F)$ and the goal is to find a permutation $\pi\colon U\to\{1,\dots,|U|\}$ that minimizes the sum of cover times of the edges $\ct\br{\pi}=\sum_{e\in F}\ct_\pi(e)$, i.e., where the \emph{cover time} of an edge is defined as $\ct_\pi(uv)=\min\brc{\pi(u),\pi(v)}$. Stankovic proves that this problem cannot be approximated within a factor of $1.0748$ assuming UGC~\cite{SomeResultsOnApproximabilityOfMinimumSumVertexCover}. The reduction is very simple. The vertex set is $V=U\cup \brc{r}$, $D=F$ and $E$ is a star with center $r$ and leaves $U$. The capacities and demands are all set to $1$.

Consider any clustering tree $\cT$. As discussed before, we can assume without loss of generality that $\cT$ is binary. Now, we observe that on stars any \emph{binary cut}, by which we mean a cut where both sides of the cut are connected, consists of some leaf $l\in U$ and nothing more. Consequently, any such clustering tree directly corresponds to a sequence of leaves, in the order they are separated from the center (i.e., according to their distances from the root in $\cT$). Conversely, every sequence $\pi\colon U\to\{1,\dots,|U|\}$ of the leaves corresponds to a unique binary clustering tree. 

Now, for any sequence $\pi\colon U\to\{1,\dots,|U|\}$ and a vertex $u\in U$, let
\[
E_\pi(u)=\{vw\in F\colon \pi(u)\leq \min\brc{\pi(v), \pi(w)}\}
\]
be the set of all edges covered no sooner than $\pi(u)$.
Then, for the clustering tree $\cT$ corresponding to $\pi$ we can rewrite its cost as 
\[
c\br{\cT}=
\sum_{e\in E} \spr{V_{\lca_\cT(e)}}=
\sum_{u\in U} \spr{V_{\lca_\cT(u,r)}}=
\sum_{u\in U}\spr{E_\pi(u)} =
\ct\br{\pi},
\]
where the third equality is due to the fact that for any vertex $u\in U$, the set $E_\pi(u)$ consists of all edges whose both endpoints appear no sooner than $\pi\br{u}$ in the sequence $\pi$, which are all of the demand edges still present in the cluster containing $u$, before cutting $u$ from the root. The fourth equality follows from basic rewriting of $\ct\br{\pi}$, since each $u\in U$ contributes one unit of cost per each edge which was not covered by the first $\pi(u)-1$ vertices in the sequence.

Therefore, the costs of corresponding solutions coincide, and the inapproximability of \textsc{Min-Sum Vertex Cover} transfers to \textsc{Hierarchical Clustering with Demands} on supply stars.
\end{proof}

\end{document}